\documentclass[journal]{IEEEtran}
\usepackage{cite}
\ifCLASSINFOpdf
\usepackage[pdftex]{graphicx}
\else
\usepackage[dvips]{graphicx}
\fi
\usepackage{cite}
\usepackage{amsmath,amssymb,amsfonts}
\usepackage{amsmath}
\usepackage{amsthm}
\usepackage{subfigure}
\usepackage{graphicx}
\usepackage{textcomp}
\usepackage{xcolor}
\usepackage{tabularx}
\usepackage{endnotes}
\usepackage{algorithm}
\usepackage{algorithmicx}
\usepackage{algpseudocode}
\usepackage{textcomp}
\usepackage{xcolor}
\usepackage{multirow}
\usepackage{algorithm}
\usepackage{algorithmicx}
\usepackage{algpseudocode}
\usepackage{mathtools}
\usepackage{mhchem}
\usepackage{tensor}
\usepackage{graphicx}  
\usepackage{psfrag}    
\usepackage{amsthm}
\usepackage{mathrsfs}
\usepackage[outdir=./]{epstopdf}
\usepackage{bm}
\usepackage{tabularx}
\newtheorem{theorem}{Theorem}
\ifCLASSOPTIONcompsoc
\usepackage[caption=false,font=normalsize,labelfont=sf,textfont=sf]{subfig}
\else
\usepackage[caption=false,font=footnotesize]{subfig}
\fi
\usepackage{fixltx2e}
\usepackage{dblfloatfix}

\usepackage{multirow}
\usepackage{booktabs}
\usepackage{url}

\begin{document}
	
	\title{Adaptive Target Detection for FDA-MIMO Radar with Training Data in Gaussian noise}
	
	\author{Ping Li,
		Bang~Huang,~\IEEEmembership{Graduate Student Member,~IEEE,}
		Wen-Qin~Wang$^*$~\IEEEmembership{Senior Member,~IEEE}
		\thanks{This work was supported by National Natural Science Foundation of
			China 62171092. (Corresponding author: Wen-Qin Wang)
		}
		\thanks{Ping, Li, Bang, Huang and Wen-Qin, Wang are with School of Information and Communication Engineering, University of Electronic Science and Technology of China, Chengdu, 611731, P. R. China. Wen-Qin Wang is also with the Yangtze Delta Region Institute (Huzhou), University of Electronic Science and Technology of China, Huzhou, 313001, P. R. China. (Emails:202122220102@std.uestc.edu.cn; huangbang@std.uestc.edu.cn; wqwang@uestc.edu.cn)}
	}

	\maketitle
	\begin{abstract}
		This paper addresses the problem of detecting a moving target embedded in Gaussian noise with an unknown covariance matrix for frequency diverse array multiple-input multiple-output (FDA-MIMO) radar. To end it, assume that obtaining a set of training data is available. Moreover, we propose three adaptive detectors in accordance with the one-step generalized likelihood ratio test (GLRT), two-step GLRT, and Rao criteria, namely OGLRT, TGLRT, and Rao. The LH adaptive matched filter (LHAMF) detector is also introduced when decomposing the Rao test. Next, all provided detectors have constant false alarm rate (CFAR) properties against the covariance matrix. Besides, the closed-form expressions for false alarm probability (PFA) and detection probability (PD) are derived. Finally, this paper substantiates the correctness of the aforementioned algorithms through numerical simulations.
	\end{abstract}
	

	%
	\IEEEpeerreviewmaketitle

	\section{Introduction}
	\IEEEPARstart{F}{requency} diverse array (FDA) radar, as introduced by Antonik et al. \cite{Antonik2006Frequency}, employs diverse transmitting carrier frequencies across adjacent elements, resulting in a range-angle-time-dependent beampattern \cite{Wang2013Range,Lang2022Lamb,Gui2022Generalized} and even a time-varying radar cross-section (RCS) \cite{Huang2022RCS}. While some radar scholars have focused on designing the transmit beampattern \cite{Liao2020Frequency,Liao2022ALow,Basit2019Adaptive}, various applications for FDA radar have emerged, such as low probability of intercept (LPI) \cite{Wang2021LPI}, lamb wave sensing \cite{Lang2022Lamb}, and others \cite{Hu2022Adaptive,Nusenu2022Power}. However, the time-varying and range-angle coupling beampattern of FDA radar make it difficult to utilize for effective target detection and localization \cite{huang2021adaptive,lan2020glrt,XuLiaoRobust17}. In response to this, scholars have proposed combining multiple-input multiple-output (MIMO) radar \cite{fishler2004mimo,JianLiMIMORadar2009} with FDA radar technology to form FDA-MIMO radar \cite{Xu2015Deceptive,XiongWangFDA18,TanMing2021Correction}. It is worth noting that FDA-MIMO radar proves effective in eliminating time-varying effects and decoupling range and angle dependencies in the beampattern \cite{TanMing2021Correction,LanLan2020Suppression,Xu2015JointRange}. Therefore, this paper primarily delves into the adaptive target detection technology enabled by FDA-MIMO radar.
	
	Moreover, \cite{Gui2022Generalized} analyzed the generalized ambiguity function (GAF) of FDA and concluded that FDA-MIMO had been endowed with the abilities, namely super range resolution and resolving the Doppler ambiguity. Further, the nature of super-resolution in the range domain has further validated the effectiveness of what FDA-MIMO radar can suppress the main-lobe deceptive jamming/clutter \cite{lan2020glrt,Xu2015JointRange,ChengJie2021Joint}. The reason is that FDA-MIMO can distinguish the echoes from different ranges \cite{lan2020glrt}. On the other side,  Wan \cite{Wan2022Resolving} utilized the characteristic of multiple frequencies of FDA-MIMO radar to resolve the Doppler ambiguity for high-speed moving target. Meanwhile, there also arise applications in other fields, such as high-resolution and wide-swath synthetic aperture radar (SAR) imaging \cite{Zhang2022High}, SAR ground moving target indication (GMTI) \cite{HuangLi2022Frequency}, deceptive jamming against SAR \cite{Huang2021FDABased} and etc \cite{Jian2022physical,Ji2019On,Wang2021Clutter}.

	The range-dependent property of FDA-MIMO radar makes obtaining homogenous training data in clutter challenging. Therefore, some studies have designed target detectors for FDA-MIMO radar without needing training data\cite{huang2021adaptive,huang2022adaptive}. Specifically, Huang \textit{et al}. \cite{huang2021adaptive} presented the detectors resorting to Rao and Wald criteria, respectively. Furthermore, leveraging the super-resolution characteristics of FDA radar in range dimension \cite{gui2021fda}, the same authors investigated the problem of distributed targets spanning one primary range cell and multiple sub-range cells \cite{Huang2022AdaptiveBa}. Nevertheless, determining the detection threshold for the detectors above can be challenging in practical scenarios. One approach is to derive the detection threshold from an extensive historical data set. 
	However, when historical data fails to represent the current environment adequately, these detectors may not function optimally.
	
	Considering an FDA-MIMO radar system operating in a noise environment, including receiver thermal noise, jammers, and clutter after secondary range compensation \cite{xu2014space, xu2015range, xu2016space}, homogeneous training data can be obtained from adjacent cells to estimate the unknown covariance matrix.  
	Hence, to detect the target embedded in Gaussian interference, Lan \textit{et al.}
	employed semidefinite programming (SDP) to determine the optimal frequency offset for detection, and  developed a one-step GLRT (OGLRT) detector \cite{lan2020glrt}. Further,
	based on the two-step GLRT (TGLRT) and model order selection (MOS) criteria, \cite{Zhu2023Target} presented a detector that can simultaneously achieve target detection and recognition.  
	Zeng \textit{et al.} \cite{Zeng2022GLRTbasedDF} investigated the moving target detection problem by considering the situations with known and unknown Doppler shifts, respectively, according to the OGLRT criterion. It is worth noting that there is an error in the analysis of the statistical characteristics of the detectors in \cite{Zeng2022GLRTbasedDF}, and this error will be rectified in this paper.  
	Meanwhile, to the best of author's knowledge, there were no relevant studies based on TGLRT and Rao criteria to study the detection problem provided in \cite{Zeng2022GLRTbasedDF}. Nevertheless, previous excellent published works  \cite{liu2014Parametric,de2007rao,de2004new,robey1992cfar,liu2019robust} indicated that, in comparison to the OGLRT, both the TGLRT and Rao detectors exhibit superior detection performance, with a notable emphasis on robustness and selectivity for mismatched signals.
	These observations have motivated us to conduct a more in-depth investigation into this issue. 
	
	
	Moreover, the detection problem investigated in this paper is a specific form of the detection model introduced in \cite{Liu2014AdaptiveHE,kelly1989adaptive}. Differently, this study's unique application of the model to the FDA-MIMO radar system distinguishes it. Furthermore, this paper goes beyond merely analysing the proposed detectors' statistical properties, presenting closed-form expressions for the probability of false alarm (PFA) and probability of detection (PD). Additionally, this study delves into the implications of Doppler dimension mismatch through thorough numerical simulation. 
	
	This paper addresses the challenge of detecting a moving target embedded in Gaussian noise with an unknown covariance matrix for FDA-MIMO radar. The main contributions are summarized as follows:
	\begin{itemize}
		\item 
		When training data is available, we design three adaptive detectors to detect a moving target embedded in an unknown Gaussian noise (includes thermal noise, jammers, and clutter after secondary range compensation)—specifically, OGLRT, TGLRT, and Rao. Additionally, an LHAMF detector is proposed when decomposing the Rao detector.
		\item  A comprehensive statistical analysis is conducted for the proposed detectors, providing closed-form expressions for the PFA and PD for four detectors. We also demonstrate their constant false alarm rate (CFAR) properties agaisnt covariance matrix under null hypothesis.  
		\item Simulation results highlight the validness of proposed methods. OGLRT and LHAMF excel in scenarios with no signal mismatch, while TGLRT demonstrates robustness, and the Rao detector exhibits optimal selectivity for mismatched signal situations. Moreover, FDA-MIMO consistently outperforms traditional MIMO in detecting target operating in the presence of Gaussian noise plus mainlobe detecptive jamming.
	\end{itemize}
	
	The subsequent structure of the paper is organized as follows. In Section \ref{sec2}, we formulate the problem of detecting a moving target as a binary hypothesis. Following that, in Section \ref{sec3}, three detectors are designed. Section \ref{sec4} provides a detailed analysis of the statistical characteristics of the proposed detectors. Section \ref{sec5} gives the simulation results, and finally, the conclusion is provided in Section \ref{sec6}.

	\textbf{\textit{Notations:}} 	
	Matrices (vectors) are defined by boldface letters. $\mathbb{C} ^{M\times N}$  is a complex matrix with dimension of $M\times N$. The symbol  $\left( \cdot \right)^{\dagger}$, $\left( \cdot \right) ^T$  and  $\left( \cdot \right)^*$ denote conjugate transpose,  transpose and complex conjugate, respectively. $\sim $ means “be distributed as”, and $C\mathcal{F}$	denotes F-distribution. $\mathcal{C} \mathcal{N} $, $CW$ and $C\beta$ represents the Gaussian, Wishart and  Beta distribution, repectively.  Hadamard product and Kronecker product are denoted by $\odot $ and $\otimes $, respectively. $\mathrm{Tr}\left[ \cdot \right] $ and $\det \left[ \cdot \right] $ represent the trace and determinant of the matrix, respectively. 
	$\ln \left( \cdot \right) $  and  $\partial \left( \cdot \right) $ denote the logarithm function and gradient operator, respectively.
	$vec\left[ \cdot \right]$ represents the vectorization operation, and $E\left( \cdot \right) $ is the statistical expectation. $\mathbf{I}$ and $\mathrm{cont}.$ denote an identity matrix with appropriate dimension and a constant. $
	\mathbf{P}_{\mathbf{B}}=\mathbf{B}\left( \mathbf{B}^{\mathbf{\dagger }}\mathbf{B} \right)^{-\mathbf{1}}\mathbf{B}^{\mathbf{\dagger }}$  is an orthogonal projection matrix onto the subspace spanned by the columns of $\mathbf{B}$, and $\mathbf{P}_{\mathbf{B}}^{\bot}=\mathbf{I}-\mathbf{P}_{\mathbf{B}}$. $\mathrm{Re}\left[ \cdot \right] $
	and $ \mathrm{Im}\left[ \cdot \right] $ represent the real and imaginary parts of a complex number, respectively. $
	\left[ \mathbf{B} \right] _k
	$ denotes the $k$th column of matrix $\mathbf{B}$.

	\section{Problem Formulation}
	\label{sec2}
	Assuming that there are $M$ FDA transmit and $N$ phased-array (PA) receive elements in a colocated FDA-MIMO radar system. Regarding the first element as the reference element, the carrier frequency of $m$th elemant can be denoted by
	\begin{equation}
		\label{1002}
		\begin{split}
			f_m=f_0+\Delta f_m, m=1,2,...,M
		\end{split},
	\end{equation}
	where $f_0$ is the carrier frequency of first elemant and $\Delta f_m=(m-1)\Delta f$ is the frequency offset between the first and $m$th elemant. Under the far-field assumption, consider a point-like target  with a radial velocity $v$, located at the distance of $r$ and the angle of $\theta $. Then, the received signal \cite{gui2017coherent, GUI2020102861, huang2021adaptive} at $k$th ($k=1,...,K$) snapshot can be expressed as 
	\begin{equation}
		\label{1003}
		\begin{split}
			\mathbf{z}_k=\xi \omega _{k}^{}\left( f_{\mathrm{d}} \right) \mathbf{a}_{\mathrm{TR}}^{}\left( r,\theta \right) +\mathbf{n}_k\in \mathbb{C} ^{MN},
		\end{split}
	\end{equation}
	where $\xi$ encompass both the target and channel effects, is unknown.  $\omega _{k}^{}\left( f_{\mathrm{d}} \right) =e^{j2\pi f_{\mathrm{d}}t}$ is the Doppler shift caused by the moving target, with $f_{\mathrm{d}}=\frac{2v}{c}f_0T_p$.  $c$ and $T_p$ represent the light speed and pulse repetition interval (PRI), respectively. Moreover,  $\mathbf{a}_{\mathrm{TR}}^{}\left( r,\theta \right) $ denotes the joint transmitting-receiving steering vector, defined as
	\begin{equation}
		\label{1004}
		\begin{split}
			\mathbf{a}_{\mathrm{TR}}^{}\left( r,\theta \right) =\mathbf{a}_{\mathrm{T}}^{}\left( r,\theta \right) \otimes \mathbf{a}_{\mathrm{R}}^{}\left( \theta \right) \in \mathbb{C} ^{MN},
		\end{split}
	\end{equation}
	where $\mathbf{a}_{\mathrm{R}}^{}\left( \theta \right) $  represents receiving steering vector, given by
	\begin{equation}
		\label{1005}
		\begin{split}
			\mathbf{a}_{\mathrm{R}}^{}\left( \theta \right) =\left[ 1, e^{j2\pi \frac{d_{\mathrm{R}}\sin \theta}{\lambda _0}},..., e^{j2\pi \left( M-1 \right) \frac{d_{\mathrm{R}}\sin \theta}{\lambda _0}} \right] \in \mathbb{C} ^N ,
		\end{split}
	\end{equation}
	and $\mathbf{a}_{\mathrm{T}}^{}\left( r,\theta \right) $ is the transmitting steering vector, with
	\begin{equation}
		\label{1006}
		\begin{split}
			\mathbf{a}_{\mathrm{T}}^{}\left( r,\theta \right) =\mathbf{a}_{t}^{}\left( \theta \right) \odot \mathbf{e}\left( -r \right) ,
		\end{split}
	\end{equation}
	where
	\begin{align}
		\label{1007}
		\begin{split}
			\mathbf{a}_{t}^{}\left(\theta \right) &=\left[ 1, e^{j2\pi \frac{d_{\mathrm{T}}\sin \theta}{\lambda _0}},..., e^{j2\pi \left( M-1 \right) \frac{d_{\mathrm{T}}\sin \theta}{\lambda _0}} \right] \in \mathbb{C} ^M, \\
			\mathbf{e}\left( r \right) &=\left[ 1,e^{j2\pi \Delta f\frac{2r}{c}},...,e^{j2\pi \left( M-1 \right) \Delta f\frac{2r}{c}} \right] \in \mathbb{C} ^M. 
		\end{split}
	\end{align}
	The symbol	$\frac{2r}{c}$ is the time delay for the signal to travel from the first transmitting element to the target and back to the first receiving element. $d_{\mathrm{T}}$ and $d_{\mathrm{R}}$ are the distance between transmit and recceive element, respectively. 
	
	Hence, with $K$ snapshots, the received signal can be expressed as
	\begin{equation}
		\label{1009}
		\begin{split}
			\mathbf{Z}=\xi \mathbf{a}_{\mathrm{TR}}^{}\left( r,\theta \right) \boldsymbol{\omega }_{\mathrm{D}}^{T}\left( f_{\mathrm{d}} \right) +\mathbf{N}\in \mathbb{C} ^{MN\times K},
		\end{split}
	\end{equation}
	where
	\begin{align}
		\label{1010}
		\begin{split}
			\mathbf{Z}&=\left[ \mathbf{z}_0,\mathbf{z}_1,\cdots ,\mathbf{z}_{K-1} \right] \in \mathbb{C} ^{MN\times K}\\
			\mathbf{N}&=\left[ \mathbf{n}_0,\mathbf{n}_1,\cdots ,\mathbf{n}_{K-1} \right] \in \mathbb{C} ^{MN\times K}\\
			\boldsymbol{\omega }_{\mathrm{D}}^{}\left( f_{\mathrm{d}} \right) &=\left[ \mathrm{\omega}_{0}^{}\left( f_{\mathrm{d}} \right) ,\mathrm{\omega}_{2}^{}\left( f_{\mathrm{d}} \right) ,\cdots ,\mathrm{\omega}_{K-1}^{}\left( f_{\mathrm{d}} \right) \right] \in \mathbb{C} ^K.
		\end{split}
	\end{align}
	
	Let $\mathcal{H} _1$ hypothesis and $\mathcal{H} _0$ hypothesis represent the presence and absence of a poin-like target in the test data, respectively. Therefore, the detection problem interested can be modeled as a binary hypothesis, given by
	\begin{equation}
		\label{eq2}
		\begin{cases}
			\mathcal{H} _1:\begin{cases}
				\mathbf{Z}=\xi \mathbf{a}_{\mathrm{TR}}^{}\left( r,\theta \right) \boldsymbol{\omega }_{\mathrm{D}}^{T}\left( f_{\mathrm{d}} \right) +\mathbf{N}\\
				\mathbf{Z}_l=\mathbf{N}_l,l=1,2,\cdots ,L\\
			\end{cases}\\
			\mathcal{H} _0:\begin{cases}
				\mathbf{Z}=\mathbf{N}\\
				\mathbf{Z}_l=\mathbf{N}_l,l=1,2,\cdots ,L\\
			\end{cases}\\
		\end{cases}
	\end{equation}
	where	$\mathbf{Z}_l=\mathbf{N}_l\in \mathbb{C} ^{MN\times K},l=1,2,\cdots ,L$ is the homogenous training data obtained from adjacent cells with $L$ being number of neighbouring bins. Each column of $\mathbf{N}$ and $\mathbf{N}_l,l=1,2,\cdots ,L$ is an independent and identically distributed (IID) zero-mean complex-valued Gaussian random vector with an unknown covariance matrix $\mathbf{R}$.
	
	Compared with the published works \cite{GUI2020102861,huang2021adaptive}, this paper will study the detection problem \eqref{eq2} with the aid of training data term $\mathbf{Z}_l,l=1,2,\cdots ,L$. Moreover, different from \cite{Zeng2022GLRTbasedDF}, this study designs additional adaptive detectors based on the TGLRT and Rao criteria. Furthermore, we will correct statistical analysis errors in \cite{Zeng2022GLRTbasedDF} for OGLRT detector. Finally, it is worth noting that the detection problem \eqref{eq2} aligns with a specific instantiation of the detection model outlined in \cite{kelly1989adaptive,Liu2014AdaptiveHE}. However, what sets this study apart is its distinctive application of the model to the FDA-MIMO radar system. Moreover, beyond a mere analysis of the statistical properties of the proposed detectors, this paper provides closed-form expressions for the PFA and PD.

	\section{Detectors design}
	\label{sec3}
	In this section, detectors will be designed based on OGLRT, TGLRT and Rao criteria, to determine the presence or absence of the moving target from received data. 
	
	Define
	\begin{equation}
		\label{eq5}
		\mathbf{Y}=\left[ \mathbf{Z}_1,\mathbf{Z}_2,\cdots ,\mathbf{Z}_L \right] \in \mathbb{C} ^{MN\times KL},
	\end{equation}
	and $\mathbf{S}=\mathbf{YY}^{\dagger}$ is the sample covariance matrix (SCM). Notably, to ensure $\mathbf{S}$ is invertible, we need $LK>MN$.
	
	The joint probability density function (PDF) of $\mathbf{Z}$ and $\mathbf{Y}$ under $\mathcal{H}_1$ can be expressed as
	\begin{equation}
		\label{eq3}
		\begin{split}&
			f\left( \mathbf{Z},\mathbf{Y}\left| \xi \right. ,\mathbf{R},\mathcal{H} _1 \right) =\frac{1}{\pi ^{MNK\left( L+1 \right)}\det ^{K\left( L+1 \right)}\left( \mathbf{R} \right)}
			\\&
			\qquad \qquad \qquad \qquad \,\,\times \exp \left\{ -\mathrm{Tr}\left[ \mathbf{R}^{-1}\mathcal{Z} _{1}^{}\mathcal{Z} _{1}^{\dagger} \right] -\mathrm{Tr}\left[ \mathbf{R}^{-1}\mathbf{S} \right] \right\}  
		\end{split}
	\end{equation}
	where
	\begin{equation}
		\label{eq4}
		\begin{split}&
			\mathcal{Z} _{1}=\mathbf{Z}-\xi \mathbf{a}_{\mathrm{TR}}\left( r,\theta \right) \boldsymbol{\omega }_{\mathrm{D}}^{T}\left( f_{\mathrm{d}} \right). 
		\end{split}
	\end{equation}
	Similarly, the joint PDF under $\mathcal{H}_0$ is
	\begin{equation}
		\begin{split}
			\label{eq6}&
			f\left( \mathbf{Z},\mathbf{Y}\left| \mathbf{R},\mathcal{H} _0 \right. \right) 
			=\frac{1}{\pi ^{MNK\left( L+1 \right)}\det ^{K\left( L+1 \right)}\left( \mathbf{R} \right)}
			\\&
			\qquad\qquad\qquad\,\,\,\,\,\,\,\, \times \exp \left\{ -\mathrm{Tr}\left[ \mathbf{R}^{-1}\mathbf{ZZ}_{}^{\dagger} \right] -\mathrm{Tr}\left[ \mathbf{R}^{-1}\mathbf{S} \right] \right\} .
		\end{split} 
	\end{equation}
	\subsection{OGLRT}
	As shown in Appendix A, the OGLRT detector can be given by
	\begin{equation}
		\label{eq200}
		\varLambda =\frac{\tilde{\mathbf{a}}^{\dagger}\left( \mathbf{I}+\tilde{\mathbf{Z}}\mathbf{P}_{\boldsymbol{\omega }_{\mathrm{D}}^{*}\left( f_{\mathrm{d}} \right)}^{\bot}\tilde{\mathbf{Z}}^{\dagger} \right) ^{-1}\tilde{\mathbf{a}}}{\tilde{\mathbf{a}}^{\dagger}\left( \mathbf{I}+\tilde{\mathbf{Z}}\tilde{\mathbf{Z}}^{\dagger} \right) ^{-1}\tilde{\mathbf{a}}}\underset{\mathcal{H} _0}{\overset{\mathcal{H} _1}{\gtrless}}\lambda.
	\end{equation}
	with
	\begin{align}
		\label{eq111}
		\tilde{\mathbf{Z}}=&\mathbf{S}^{-1/2}\mathbf{Z},\\
		\label{eq121}
		\tilde{\mathbf{a}}=&\mathbf{S}^{-1/2}\mathbf{a}_{\mathrm{TR}}\left( r,\theta \right) .
	\end{align}
	
	It is essential  to highlight that \eqref{eq200} is identical to the detector Eq.(24) in \cite{Zeng2022GLRTbasedDF}. However, the authors believe there are some errors in the statistical analysis conducted by \cite{Zeng2022GLRTbasedDF}. These concerns will be further verified by the theoretical derivation and numerical simulations in Section \ref{sec41} and Section \ref{sec5}. Additionally, regarding structure, OGLRT resembles Kelly's GLRT (KGLRT)\cite{kelly1986adaptive}, and the distinctions between them will also be detailed in Section \ref{sec41}.
	
	\subsection{TGLRT}
	Suppose the covariance matrix $\mathbf{R}$ is known in advance,	the decision statistic of TGLRT criterion \cite{robey1992cfar} can be written as
	\begin{equation}
		\label{eq21}
		\varLambda_{TGLRT} =\frac{\underset{\xi}{\max}f\left( \mathbf{Z},\mathbf{Y}\left| \xi \right. ,\mathbf{R},\mathcal{H} _1 \right)}{\max f\left( \mathbf{Z},\mathbf{Y}\left| \mathbf{R},\mathcal{H} _0 \right. \right)}\underset{\mathcal{H} _0}{\overset{\mathcal{H} _1}{\gtrless}}\lambda. 
	\end{equation}
	Inserting \eqref{eq3} and \eqref{eq6} into \eqref{eq21} results in
	\begin{equation}
		\label{eq22}
		\varLambda_{TGLRT} =\underset{\xi}{\max}e^{\left\{ \mathrm{Tr}\left[ \mathbf{R}^{-1}\mathbf{ZZ}_{}^{\dagger} \right] -\mathrm{Tr}\left[ \mathbf{R}^{-1}\mathcal{Z} _1\mathcal{Z} _{1}^{\dagger} \right] \right.}\}\underset{\mathcal{H} _0}{\overset{\mathcal{H} _1}{\gtrless}}\lambda.
	\end{equation}
	Then, define
	\begin{equation}
		\label{eq23}
		g\left( \xi \right) =\mathrm{Tr}\left[ \mathbf{R}^{-1}\mathcal{Z} _1\mathcal{Z} _{1}^{\dagger} \right] ,
	\end{equation}
	Derivative $g(\xi)$ with respect to (w.r.t.) $\xi$, yields  
	\begin{equation}
		\label{eq24}
		\frac{\partial g\left( \xi \right)}{\partial \xi}=\boldsymbol{\omega }_{\mathrm{D}}^{T}\left( f_{\mathrm{d}} \right) \left( \mathbf{Z}-\xi \mathbf{a}_{\mathrm{TR}}\left( r,\theta \right) \boldsymbol{\omega }_{\mathrm{D}}^{T}\left( f_{\mathrm{d}} \right) \right) ^{\dagger}\mathbf{R}^{-1}\mathbf{a}_{\mathrm{TR}}\left( r,\theta \right). 
	\end{equation}
	Equating  $\frac{\partial g\left( \xi \right)}{\partial \xi}$ to zero gives the maximum likelihood estimation (MLE) of $\xi$, as
	\begin{equation}
		\label{eq25}
		\hat{\xi}=\frac{\mathbf{a}_{\mathrm{TR}}^{\dagger}\left( r,\theta \right) \mathbf{R}^{-1}\mathbf{Z\omega }_{\mathrm{D}}^{*}\left( f_{\mathrm{d}} \right)}{\mathbf{a}_{\mathrm{TR}}^{\dagger}\left( r,\theta \right) \mathbf{R}^{-1}\mathbf{a}_{\mathrm{TR}}\left( r,\theta \right) \boldsymbol{\omega }_{\mathrm{D}}^{T}\left( f_{\mathrm{d}} \right) \boldsymbol{\omega }_{\mathrm{D}}^{*}\left( f_{\mathrm{d}} \right)}.
	\end{equation}
	Plugging \eqref{eq25} into \eqref{eq23}, leads to
	\begin{equation}
		\begin{split}
			\label{eq26}&
			g\left( \hat{\xi} \right) =\mathrm{Tr}\left[ \mathbf{Z}^{\dagger}\mathbf{R}^{-1}\mathbf{Z}-\hat{\xi}\boldsymbol{\omega }_{\mathrm{D}}^{T}\left( f_{\mathrm{d}} \right) \mathbf{Z}^{\dagger}\mathbf{R}^{-1}\mathbf{a}_{\mathrm{TR}}\left( r,\theta \right) \right]. 
		\end{split}
	\end{equation}
	Substituting \eqref{eq26} into \eqref{eq22} and taking the logarithm of the result yield the expression for the TGLRT detector for known $\mathbf{R}$, given by
	\begin{equation}
		\begin{split}
			\label{eq27}&
			\varLambda _{TGLRT}=\frac{\mathbf{a}_{\mathrm{TR}}^{\dagger}\left( r,\theta \right) \mathbf{R}^{-1}\mathbf{Z}\boldsymbol{\omega }_{\mathrm{D}}^{*}\left( f_{\mathrm{d}} \right) \boldsymbol{\omega }_{\mathrm{D}}^{T}\left( f_{\mathrm{d}} \right) \mathbf{Z}^{\dagger}\mathbf{R}^{-1}\mathbf{a}_{\mathrm{TR}}\left( r,\theta \right)}{\mathbf{a}_{\mathrm{TR}}^{\dagger}\left( r,\theta \right) \mathbf{R}^{-1}\mathbf{a}_{\mathrm{TR}}\left( r,\theta \right) \boldsymbol{\omega }_{\mathrm{D}}^{T}\left( f_{\mathrm{d}} \right) \boldsymbol{\omega }_{\mathrm{D}}^{*}\left( f_{\mathrm{d}} \right)}
			\\&
			\qquad\quad\,\,\,\,       \underset{\mathcal{H} _0}{\overset{\mathcal{H} _1}{\gtrless}}\lambda . 
		\end{split}
	\end{equation}
	
	However, in practice, $\mathbf{R}$ is so unknown that \eqref{eq27} cannot be directly used. As a consequence, replacing $\mathbf{R}$ in \eqref{eq27} with SCM $\mathbf{S}$, the final TGLRT detector for the problem interested can be given by
	\begin{equation}
		\label{eq28}
		\varLambda_{TGLRT} =\frac{\left| \tilde{\mathbf{a}}_{}^{\dagger}\tilde{\mathbf{Z}}_{\boldsymbol{\omega }} \right|^2}{\tilde{\mathbf{a}}_{}^{\dagger}\tilde{\mathbf{a}}}\underset{\mathcal{H} _0}{\overset{\mathcal{H} _1}{\gtrless}}\lambda ,
	\end{equation}
	where 
	\begin{equation}
		\label{eq29}
		\tilde{\mathbf{Z}}_{\boldsymbol{\omega }}=\tilde{\mathbf{Z}}\frac{\boldsymbol{\omega }_{\mathrm{D}}^{*}\left( f_{\mathrm{d}} \right)}{\sqrt{\boldsymbol{\omega }_{\mathrm{D}}^{T}\left( f_{\mathrm{d}} \right) \boldsymbol{\omega }_{\mathrm{D}}^{*}\left( f_{\mathrm{d}} \right)}}.
	\end{equation}
	If we consider $\tilde{\mathbf{Z}}_{\boldsymbol{\omega }}$ as the received data, \eqref{eq28} shares the same form as the detector proposed by Robey \textit{et al}. \cite{robey1992cfar}. However, the critical distinction is that the received data in \eqref{eq29} also includes the Doppler information of the moving target. Given the established equivalence between TGLRT and the Wald test \cite{de2004new}, we avoid designing detectors based on the Wald criterion in this paper.
	
	\subsection{Rao}
	Assuming that $\mathbf{\Theta}$ is a parameter vector,  defined as
	\begin{equation}
		\label{eq141}
		\mathbf{\Theta }=\left[ \mathbf{\Theta }_{\mathrm{r}}^{T},\mathbf{\Theta }_{\mathrm{s}}^{T} \right] ^T,
	\end{equation}
	where,  $\mathbf{\Theta }_{\mathrm{r}}^{}=\left[ \xi _{\mathrm{R}},\xi _{\mathrm{I}} \right] $ is useful parameter vector, with $\xi _{\mathrm{R}}$ and $\xi _{\mathrm{I}}$ represent real and imaginary parts of $ \xi $, respectively. 
	$\mathbf{\Theta }_{\mathrm{s}}^{}=vec\left( \mathbf{R} \right) $ is an extra parameter vector.
	Furthermore, the Fisher information matrix (FIM) can be presented as
	\begin{equation}
		\label{eq143}
		\mathbf{F}\left( \mathbf{\Theta } \right) =E\left[ \left( \frac{\partial \ln f\left( \mathbf{Z}\left| \xi \right. ,\mathcal{H} _1 \right)}{\partial \mathbf{\Theta }^*} \right) \left( \frac{\partial \ln f\left( \mathbf{Z}\left| \xi \right. ,\mathcal{H} _1 \right)}{\partial \mathbf{\Theta }^T} \right) \right] ,
	\end{equation}
	or partitioned as
	\begin{equation}
		\label{eq144}
		\mathbf{F}\left( \mathbf{\Theta } \right) =\left[ \begin{array}{c}
			\mathbf{F}_{\mathbf{\Theta }_{\mathrm{r}},\mathbf{\Theta }_{\mathrm{r}}}\left( \mathbf{\Theta } \right) ,\mathbf{F}_{\mathbf{\Theta }_{\mathrm{r}},\mathbf{\Theta }_{\mathrm{s}}}\left( \mathbf{\Theta } \right)\\
			\mathbf{F}_{\mathbf{\Theta }_{\mathrm{s}},\mathbf{\Theta }_{\mathrm{r}}}\left( \mathbf{\Theta } \right) ,\mathbf{F}_{\mathbf{\Theta }_{\mathrm{s}},\mathbf{\Theta }_{\mathrm{s}}}\left( \mathbf{\Theta } \right)\\
		\end{array} \right] .
	\end{equation}
	Next, the statistic of Rao test \cite{kay1993fundamentals} can be expressed as
	\begin{equation}
		\begin{split}&
			\label{eq145}
			\varLambda_{Rao} =\left. \frac{\partial \ln f\left( \mathbf{Z}\left| \xi \right. ,\mathcal{H} _1 \right)}{\partial \mathbf{\Theta }_{\mathrm{r}}} \right|_{\mathbf{\Theta }=\hat{\mathbf{\Theta}}_0}^{T}\left[ \mathbf{F}^{-1}\left( \hat{\mathbf{\Theta}}_0 \right) \right] _{\mathbf{\Theta }_{\mathrm{r}},\mathbf{\Theta }_{\mathrm{r}}}
			\\&
			\qquad\qquad\qquad\qquad    \left. \,\,\times \frac{\partial \ln f\left( \mathbf{Z}\left| \xi \right. ,\mathcal{H} _1 \right)}{\partial \mathbf{\Theta }_{\mathrm{r}}^{*}} \right|_{\mathbf{\Theta }=\hat{\mathbf{\Theta}}_0}^{}\underset{\mathcal{H} _0}{\overset{\mathcal{H} _1}{\gtrless}}\lambda,
		\end{split}
	\end{equation}
	where $\hat{\mathbf{\Theta}}_0=\left[ \hat{\mathbf{\Theta}}_{\mathrm{r}_0}^{T},\hat{\mathbf{\Theta}}_{\mathrm{s}_0}^{T} \right] ^T$ is the MLE of $\mathbf{\Theta }$ under $\mathcal{H} _0$, and 
	\begin{align} 
		\label{eq30}
		\hat{\mathbf{\Theta}}_{r_0}^{}&=\left[ 0, 0 \right] \\
		\hat{\mathbf{\Theta}}_{\mathrm{s}_0}^{}&=E\left[ vec\left( \mathbf{ZZ}^{\dagger}+\mathbf{S} \right) \right] 
	\end{align} 
	In addition, we have
	\begin{equation}
		\begin{split}
			\label{eq146}
			\frac{\partial \ln f\left( \mathbf{Z}\left| \xi \right. ,\mathcal{H} _1 \right)}{\partial \mathbf{\Theta }_{\mathrm{r}}}
			=\left[ \frac{\partial \ln f\left( \mathbf{Z}\left| \xi \right. ,\mathcal{H} _1 \right)}{\partial \xi _{\mathrm{R}}},\frac{\partial \ln f\left( \mathbf{Z}\left| \xi \right. ,\mathcal{H} _1 \right)}{\partial \xi _{\mathrm{I}}} \right] ^T,
		\end{split}
	\end{equation}
	with
	\begin{equation}
		\begin{split}&
			\frac{\partial \ln f\left( \mathbf{Z}\left| \xi \right. ,\mathcal{H} _1 \right)}{\partial \xi _{\mathrm{R}}}=2\mathrm{Re}\left[ \boldsymbol{\omega }_{\mathrm{D}}^{T}\left( f_{\mathrm{d}} \right) \mathcal{Z}_{1}^{\dagger}\mathbf{R}^{-1}\mathbf{a}_{\mathrm{TR}}\left( r,\theta \right) \right],  
			%
		\end{split}
	\end{equation}
	\begin{equation}
		\begin{split}&
			\frac{\partial \ln f\left( \mathbf{Z}\left| \xi \right. ,\mathcal{H} _1 \right)}{\partial \xi _{\mathrm{I}}}=2\mathrm{Im}\left[ \boldsymbol{\omega }_{\mathrm{D}}^{T}\left( f_{\mathrm{d}} \right) \mathcal{Z}_{1}^{\dagger}\mathbf{R}^{-1}\mathbf{a}_{\mathrm{TR}}\left( r,\theta \right) \right] .  
		\end{split}
	\end{equation}
	Moreover, by resorting to Schuler complementary decomposition theorem\cite{zhang2017matrix}, $\left[ \mathbf{F}^{-1}\left( \hat{\mathbf{\Theta}}_0 \right) \right] _{\mathbf{\Theta }_{\mathrm{r}},\mathbf{\Theta }_{\mathrm{r}}}$ can be recast as
	\begin{equation}\label{147}
		\begin{split}&
			\left[ \mathbf{F}^{-1}\left( \hat{\mathbf{\Theta}}_0 \right) \right] _{\mathbf{\Theta }_{\mathrm{r}},\mathbf{\Theta }_{\mathrm{r}}}=\left[ \mathbf{F}_{\mathbf{\Theta }_{\mathrm{r}},\mathbf{\Theta }_{\mathrm{r}}}\left( \mathbf{\Theta } \right) \right. 
			\\&
			\qquad\left.  -\mathbf{F}_{\mathbf{\Theta }_{\mathrm{r}},\mathbf{\Theta }_{\mathrm{s}}}\left( \mathbf{\Theta } \right) \mathbf{F}_{\mathbf{\Theta }_{\mathrm{s}},\mathbf{\Theta }_{\mathrm{s}}}^{-1}\left( \mathbf{\Theta } \right) \mathbf{F}_{\mathbf{\Theta }_{\mathrm{s}},\mathbf{\Theta }_{\mathrm{r}}}\left( \mathbf{\Theta } \right) \right] ^{-1}\left| _{\mathbf{\Theta }=\hat{\mathbf{\Theta}}_0} \right. .
		\end{split}
	\end{equation}
	In this respect,	$\mathbf{F}_{\mathbf{\Theta }_{\mathrm{r}},\mathbf{\Theta }_{\mathrm{r}}}\left( \mathbf{\Theta } \right) $ can be partitioned as
	\begin{equation}
		\begin{split}
			\label{148}&
			\mathbf{F}_{\mathbf{\Theta }_{\mathrm{r}},\mathbf{\Theta }_{\mathrm{r}}}\left( \mathbf{\Theta } \right) =-E\left[ \begin{array}{c}
				\frac{\partial ^2\ln f\left( \mathbf{Z}\left| \xi \right. ,\mathcal{H} _1 \right)}{\partial \xi _{R}^{2}},\frac{\partial ^2\ln f\left( \mathbf{Z}\left| \xi \right. ,\mathcal{H} _1 \right)}{\partial \xi _{\mathrm{R}}\partial \xi _{\mathrm{I}}}\\
				\frac{\partial ^2\ln f\left( \mathbf{Z}\left| \xi \right. ,\mathcal{H} _1 \right)}{\partial \xi _{\mathrm{I}}\partial \xi _{\mathrm{R}}},\frac{\partial ^2\ln f\left( \mathbf{Z}\left| \xi \right. ,\mathcal{H} _1 \right)}{\partial \xi _{I}^{2}}\\
			\end{array} \right] 
		\end{split}
	\end{equation}
	where
	\begin{align}
		\begin{split}
			\label{149}&
			\frac{\partial ^2\ln f\left( \mathbf{Z}\left| \xi \right. ,\mathcal{H} _1 \right)}{\partial \xi _{\mathrm{R}}^{2}}=-2\boldsymbol{\omega }_{\mathrm{D}}^{T}\left( f_{\mathrm{d}} \right) \boldsymbol{\omega }_{\mathrm{D}}^{*}\left( f_{\mathrm{d}} \right) 
			\\&
			\qquad\qquad\qquad\qquad \,\,\,\,       \times \mathbf{a}_{\mathrm{TR}}^{\dagger}\left( r,\theta \right) \mathbf{R}^{-1}\mathbf{a}_{\mathrm{TR}}\left( r,\theta \right) ,
			\\&
			\frac{\partial ^2\ln f\left( \mathbf{Z}\left| \xi \right. ,\mathcal{H} _1 \right)}{\partial \xi _{\mathrm{I}}^{2}}=-2\boldsymbol{\omega }_{\mathrm{D}}^{T}\left( f_{\mathrm{d}} \right) \boldsymbol{\omega }_{\mathrm{D}}^{*}\left( f_{\mathrm{d}} \right) 
			\\&
			\qquad\qquad\qquad\qquad \,\,\,\,     \times \mathbf{a}_{\mathrm{TR}}^{\dagger}\left( r,\theta \right) \mathbf{R}^{-1}\mathbf{a}_{\mathrm{TR}}\left( r,\theta \right) ,
		\end{split}
	\end{align}
	\begin{align}
		\begin{split}
			\frac{\partial ^2\ln f\left( \mathbf{Z}\left| \xi \right. ,\mathcal{H} _1 \right)}{\partial \xi _{\mathrm{I}}\partial \xi _{\mathrm{R}}}=\frac{\partial ^2\ln f\left( \mathbf{Z}\left| \xi \right.,\mathcal{H} _1 \right)}{\partial \xi _{\mathrm{R}}\partial \xi _{\mathrm{I}}}=0.
		\end{split}
	\end{align}
	One can prove that $\mathbf{F}_{\mathbf{\Theta }_{\mathrm{r}},\mathbf{\Theta }_{\mathrm{s}}}\left( \mathbf{\Theta } \right) =\mathbf{0}_{2,\mathrm{MN}}$. Hence,  \eqref{147} can be recast as
	\begin{equation}
		\begin{split}
			\label{151}&
			\left[ \mathbf{F}^{-1}\left( \hat{\mathbf{\Theta}}_0 \right) \right] _{\mathbf{\Theta }_{\mathrm{r}},\mathbf{\Theta }_{\mathrm{r}}}=\left[ \mathbf{F}_{\mathbf{\Theta }_{\mathrm{r}},\mathbf{\Theta }_{\mathrm{r}}}\left( \mathbf{\Theta } \right) \right] ^{-1}\left| _{\mathbf{\Theta }=\hat{\mathbf{\Theta}}_0} \right.. 
		\end{split}
	\end{equation}
	Next, taking above derived results into \eqref{eq145} gives the Rao rest, as
	\begin{equation}
		\label{eq31}
		\begin{split}&
			\varLambda_{Rao} =\frac{\left| \mathbf{a}_{\mathrm{TR}}^{\dagger}\left( r,\theta \right) \left( \mathbf{ZZ}^{\dagger}+\mathbf{S} \right) ^{-1}\mathbf{Z}\boldsymbol{\omega }_{\mathrm{D}}^{*}\left( f_{\mathrm{d}} \right) \right|^2}{\mathbf{a}_{\mathrm{TR}}^{\dagger}\left( r,\theta \right) \left( \mathbf{ZZ}^{\dagger}+\mathbf{S} \right) ^{-1}\mathbf{a}_{\mathrm{TR}}\left( r,\theta \right) \boldsymbol{\omega }_{\mathrm{D}}^{T}\left( f_{\mathrm{d}} \right) \boldsymbol{\omega }_{\mathrm{D}}^{*}\left( f_{\mathrm{d}} \right)}
			\\&
			\,\,\qquad \underset{\mathcal{H} _0}{\overset{\mathcal{H} _1}{\gtrless}}\lambda, 
		\end{split}
	\end{equation}
	or equivalently
	\begin{equation}
		\label{eq32}
		\varLambda_{Rao} =\frac{\left| \tilde{\mathbf{a}}_{}^{\dagger}\left( \mathbf{I}+\tilde{\mathbf{Z}}\tilde{\mathbf{Z}}^{\dagger} \right) ^{-1}\tilde{\mathbf{Z}}_{\boldsymbol{\omega }} \right|^2}{\tilde{\mathbf{a}}_{}^{\dagger}\left( \mathbf{I}+\tilde{\mathbf{Z}}\tilde{\mathbf{Z}}^{\dagger} \right) ^{-1}\tilde{\mathbf{a}}}\underset{\mathcal{H} _0}{\overset{\mathcal{H} _1}{\gtrless}}\lambda. 
	\end{equation}
	Although the expression for the Rao detector proposed in \cite{de2007rao} is similar to \eqref{eq32}, the key difference lies in the inverse matrix term $\left( \mathbf{I}+\tilde{\mathbf{Z}}\tilde{\mathbf{Z}}^{\dagger} \right) $ when $\tilde{\mathbf{Z}}_{\boldsymbol{\omega }}$ is considered as the received data.
	%
	%
	
	\section{Statistical Analysis}
	In this section, we will derive closed-form expressions for the PFA and PD of the proposed detectors to facilitate their performance evaluation. Furthermore, during the analysis of Rao detector, the LHAMF detector will be introduced.
	\label{sec4}
	\subsection{The PFA and PD of  OGLRT}
	\label{sec41}
	According to the matrix inversion lemma \cite{zhang2017matrix}, we have
	\begin{equation}
		\label{eq33}
		\begin{split}&
			\left( \mathbf{I}+\tilde{\mathbf{Z}}\tilde{\mathbf{Z}}^{\dagger} \right) ^{-1}
			\\&
			=\left( \mathbf{I}+\tilde{\mathbf{Z}}\mathbf{P}_{\boldsymbol{\omega }_{\mathrm{D}}^{*}\left( f_{\mathrm{d}} \right)}^{\bot}\tilde{\mathbf{Z}}^{\dagger}+\tilde{\mathbf{Z}}_{\boldsymbol{\omega }}\tilde{\mathbf{Z}}_{\boldsymbol{\omega }}^{\dagger} \right) ^{-1}
			\\&
			=\left( \mathbf{I}+\tilde{\mathbf{Z}}\mathbf{P}_{\boldsymbol{\omega }_{\mathrm{D}}^{*}\left( f_{\mathrm{d}} \right)}^{\bot}\tilde{\mathbf{Z}}^{\dagger} \right) ^{-1}
			\\&
			\,\,\,\,-\frac{\left( \mathbf{I}+\tilde{\mathbf{Z}}\mathbf{P}_{\boldsymbol{\omega }_{\mathrm{D}}^{*}\left( f_{\mathrm{d}} \right)}^{\bot}\tilde{\mathbf{Z}}^{\dagger} \right) ^{-1}\tilde{\mathbf{Z}}_{\boldsymbol{\omega }}\tilde{\mathbf{Z}}_{\boldsymbol{\omega }}^{\dagger}\left( \mathbf{I}+\tilde{\mathbf{Z}}\mathbf{P}_{\boldsymbol{\omega }_{\mathrm{D}}^{*}\left( f_{\mathrm{d}} \right)}^{\bot}\tilde{\mathbf{Z}}^{\dagger} \right) ^{-1}}{1+\tilde{\mathbf{Z}}_{\boldsymbol{\omega }}^{\dagger}\left( \mathbf{I}+\tilde{\mathbf{Z}}\mathbf{P}_{\boldsymbol{\omega }_{\mathrm{D}}^{*}\left( f_{\mathrm{d}} \right)}^{\bot}\tilde{\mathbf{Z}}^{\dagger} \right) ^{-1}\tilde{\mathbf{Z}}_{\boldsymbol{\omega}}}.
		\end{split}
	\end{equation}
	Taking \eqref{eq33} into \eqref{eq200}, the expression of OGLRT can be recast as
	\begin{equation}
		\label{eq35}
		\varLambda =\frac{1}{1-\varLambda ^\prime}\underset{\mathcal{H} _0}{\overset{\mathcal{H} _1}{\gtrless}}\lambda, 
	\end{equation}
	where 
	\begin{equation}
		\begin{split}
			\label{key}&
			\varLambda ^{\prime}=\frac{\left| \tilde{\mathbf{a}}^{\dagger}\left( \mathbf{I}+\tilde{\mathbf{Z}}\mathbf{P}_{\boldsymbol{\omega }_{\mathrm{D}}^{*}\left( f_{\mathrm{d}} \right)}^{\bot}\tilde{\mathbf{Z}}^{\dagger} \right) ^{-1}\tilde{\mathbf{Z}}_{\boldsymbol{\omega }}^{} \right|^2}{\tilde{\mathbf{a}}^{\dagger}\left( \mathbf{I}+\tilde{\mathbf{Z}}\mathbf{P}_{\boldsymbol{\omega }_{\mathrm{D}}^{*}\left( f_{\mathrm{d}} \right)}^{\bot}\tilde{\mathbf{Z}}^{\dagger} \right) ^{-1}\tilde{\mathbf{a}}}
			\\&
			\qquad
			\times \frac{1}{1+\tilde{\mathbf{Z}}_{\boldsymbol{\omega }}^{\dagger}\left( \mathbf{I}+\tilde{\mathbf{Z}}\mathbf{P}_{\boldsymbol{\omega }_{\mathrm{D}}^{*}\left( f_{\mathrm{d}} \right)}^{\bot}\tilde{\mathbf{Z}}^{\dagger} \right) ^{-1}\tilde{\mathbf{Z}}_{\boldsymbol{\omega }}^{}}.
		\end{split}
	\end{equation}
	Furthermore, we define \eqref{eq36} as it appears on the first line of the next page. 
	\begin{figure*}[htp]
		\begin{equation}
			\label{eq36}
			\underset{\text{\underline{\hspace{17cm}}}}{
				\begin{split}&
					\varLambda ^{''}=\frac{\varLambda ^\prime}{1-\varLambda ^\prime}
					\\&
					\,\,\,\,\,\,\,=\frac{\left| \tilde{\mathbf{a}}^{\dagger}\left( \mathbf{I}+\tilde{\mathbf{Z}}\mathbf{P}_{\boldsymbol{\omega }_{\mathrm{D}}^{*}\left( f_{\mathrm{d}} \right)}^{\bot}\tilde{\mathbf{Z}}^{\dagger} \right) ^{-1}\tilde{\mathbf{Z}}_{\boldsymbol{\omega }}^{} \right|^2/\tilde{\mathbf{a}}^{\dagger}\left( \mathbf{I}+\tilde{\mathbf{Z}}\mathbf{P}_{\boldsymbol{\omega }_{\mathrm{D}}^{*}\left( f_{\mathrm{d}} \right)}^{\bot}\tilde{\mathbf{Z}}^{\dagger} \right) ^{-1}\tilde{\mathbf{a}}}{1+\tilde{\mathbf{Z}}_{\boldsymbol{\omega }}^{\dagger}\left( \mathbf{I}+\tilde{\mathbf{Z}}\mathbf{P}_{\boldsymbol{\omega }_{\mathrm{D}}^{*}\left( f_{\mathrm{d}} \right)}^{\bot}\tilde{\mathbf{Z}}^{\dagger} \right) ^{-1}\tilde{\mathbf{Z}}_{\boldsymbol{\omega }}^{}-\frac{\left| \tilde{\mathbf{a}}^{\dagger}\left( \mathbf{I}+\tilde{\mathbf{Z}}\mathbf{P}_{\boldsymbol{\omega }_{\mathrm{D}}^{*}\left( f_{\mathrm{d}} \right)}^{\bot}\tilde{\mathbf{Z}}^{\dagger} \right) ^{-1}\tilde{\mathbf{Z}}_{\boldsymbol{\omega }}^{} \right|^2}{\tilde{\mathbf{a}}^{\dagger}\left( \mathbf{I}+\tilde{\mathbf{Z}}\mathbf{P}_{\boldsymbol{\omega }_{\mathrm{D}}^{*}\left( f_{\mathrm{d}} \right)}^{\bot}\tilde{\mathbf{Z}}^{\dagger} \right) ^{-1}\tilde{\mathbf{a}}}},
			\end{split}}
		\end{equation}	     
	\end{figure*}  

	Since three expressions $\varLambda$, $\varLambda ^{'}$ and $\varLambda ^{''}$ are equivalent \cite{DeMaioGreco2015ModernRadar}, the subsequent analysis in this paper primarily focuses on \eqref{eq36}. If we let  $\tilde{\mathbf{a}}$, $\tilde{\mathbf{Z}}_{\boldsymbol{\omega }}$, and $
	\mathbf{I}+\tilde{\mathbf{Z}}\mathbf{P}_{\boldsymbol{\omega }_{\mathrm{D}}^{*}\left( f_{\mathrm{d}} \right)}^{\bot}\tilde{\mathbf{Z}}^{\dagger}$ denote the steering vector, received data, and SCM, respectively, the detector represented by \eqref{eq36} is formally consistent with the KGLRT. In addition, the term $	\tilde{\mathbf{a}}^{\dagger}\left( \mathbf{I}+\tilde{\mathbf{Z}}\mathbf{P}_{\boldsymbol{\omega }_{\mathrm{D}}^{*}\left( f_{\mathrm{d}} \right)}^{\bot}\tilde{\mathbf{Z}}^{\dagger} \right) ^{-1}\tilde{\mathbf{Z}}_{\boldsymbol{\omega }}^{}$ can be further rewritten as 
	\begin{equation}
		\label{eq377}
		\tilde{\mathbf{a}}^{\dagger}\left( \mathbf{I}+\tilde{\mathbf{Z}}\mathbf{P}_{\boldsymbol{\omega }_{\mathrm{D}}^{*}\left( f_{\mathrm{d}} \right)}^{\bot}\tilde{\mathbf{Z}}^{\dagger} \right) ^{-1}\tilde{\mathbf{Z}}_{\boldsymbol{\omega }}^{}=\frac{\mathbf{a}_{\mathrm{TR}}^{\dagger}\left( r,\theta \right) \mathbf{S}_{+}^{-1}\mathbf{Z}\boldsymbol{\omega }_{\mathrm{D}}^{*}\left( f_{\mathrm{d}} \right)}{\sqrt{\boldsymbol{\omega }_{\mathrm{D}}^{T}\left( f_{\mathrm{d}} \right) \boldsymbol{\omega }_{\mathrm{D}}^{*}\left( f_{\mathrm{d}} \right)}},
	\end{equation}
	with
	\begin{equation}
		\label{eq3800}
		\mathbf{S}_{+}^{}=\mathbf{S}+\mathbf{Z}\mathbf{P}_{\boldsymbol{\omega }_{\mathrm{D}}^{*}\left( f_{\mathrm{d}} \right)}^{\bot}\mathbf{Z}^{\dagger}.
	\end{equation}
	
	Next, our focus can turn on the term $\mathbf{S}_{+}$. 
	\begin{theorem}
		Under $\mathcal{H} _0$ hypothesis, $\mathbf{S}_{+}$ obeys a complex Wishart distribution with degrees of freedom (DoFs) being $\left( L+1 \right) K$, and $MN$ and covariance matrix being $\mathbf{R}$.
	\end{theorem}	
	\begin{proof}
		One can easily prove that \cite{bandiera2022advanced} 
		\begin{equation}
			\label{eq399}
			\mathbf{S}\sim CW\left( LK,MN;\mathbf{R} \right). 
		\end{equation}
		
		Furthermore, based on the results from Appendix A in  \cite{kelly1989adaptive}, each column of $\mathbf{ZP}_{\boldsymbol{\omega }_{\mathrm{D}}^{*}\left( f_{\mathrm{d}} \right)}^{\bot}
		\in \mathbb{C} ^{MN\times K}
		$ is zero-mean Gaussian vector with covariance matrix $\mathbf{R}$, namely
		$
		\left[ \mathbf{ZP}_{\boldsymbol{\omega }_{\mathrm{D}}^{*}\left( f_{\mathrm{d}} \right)}^{\bot} \right] _k\sim \mathcal{C} \mathcal{N} \left( 0,\mathbf{R} \right) ,k=1,2,...,K
		$ under $\mathcal{H} _0$ hypothesis.
		
		 Due to $
		\mathbf{ZP}_{\boldsymbol{\omega }_{\mathrm{D}}^{*}\left( f_{\mathrm{d}} \right)}^{\bot}\left( \mathbf{P}_{\boldsymbol{\omega }_{\mathrm{D}}^{*}\left( f_{\mathrm{d}} \right)}^{\bot} \right) ^{\dagger}\mathbf{Z}^{\dagger}=\mathbf{ZP}_{\boldsymbol{\omega }_{\mathrm{D}}^{*}\left( f_{\mathrm{d}} \right)}^{\bot}\mathbf{Z}^{\dagger}
		$, we have
		\begin{equation}
			\label{eq400}
			\mathbf{Z}\mathbf{P}_{\boldsymbol{\omega }_{\mathrm{D}}^{*}\left( f_{\mathrm{d}} \right)}^{\bot}\mathbf{Z}^{\dagger}\sim CW\left( K,MN;\mathbf{R} \right). 
		\end{equation}
		By resorting to the result of [Theorem 7.3.2] \cite{anderson1984introduction}, one can  easily obtain that 
		\begin{equation}
			\label{eq401}
			\mathbf{S}_{+}^{}\sim CW\left( \left( L+1 \right) K,MN;\mathbf{R} \right). 
		\end{equation}
	\end{proof}	
	Despite \cite{Zeng2022GLRTbasedDF} stating that $\mathbf{S}_{+}^{}$ follows a complex Wishart distribution, the DoFs provided are incorrect.
	
	Based on the analysis above, we can further derive that
	\begin{equation}
		\label{eq37}
		\varLambda ^{''}\sim \begin{cases}
			C\mathcal{F} _{1,\mathrm{MM}-1}\\
			C\mathcal{F} _{1,\mathrm{MM}-1}\left( \alpha \mathscr{B} \right)\\
		\end{cases}\begin{array}{c}
			\mathrm{under}\,\,\mathcal{H} _0,\\
			\mathrm{under}\,\,\mathcal{H} _1,\\
		\end{array}
	\end{equation}
	where  $\mathrm{MM}=\left( L+1 \right) K-MN+1$ and $\alpha \mathscr{B} $ denote the DoFs and the non-central parameter, respectively, with
	\begin{equation}
		\label{388}
		\alpha =\left| \xi \right|^2\boldsymbol{\omega }_{\mathrm{D}}^{T}\left( f_{\mathrm{d}} \right) \boldsymbol{\omega }_{\mathrm{D}}^{*}\left( f_{\mathrm{d}} \right) \mathbf{a}_{\mathrm{TR}}^{\dagger}\left( r,\theta \right) \mathbf{R}^{-1}\mathbf{a}_{\mathrm{TR}}^{}\left( r,\theta \right) .
	\end{equation}
	Moreover, $\mathscr{B}$ follows the central complex Beta distribution, i.e.,
	\begin{equation}
		\label{39}
		\mathscr{B} \sim \begin{cases}
			C\beta _{\mathrm{MM}+1,MN-1}\\
			C\beta _{\mathrm{MM}+1,MN-1}\\
		\end{cases}\begin{array}{c}
			\mathrm{under}\,\,\mathcal{H} _0\\
			\mathrm{under}\,\,\mathcal{H} _1\\
		\end{array}.
	\end{equation}
	Further, the PDF of $\mathscr{B}$ is
	\begin{equation}
		\label{45}
		\begin{split}
			f_{\mathscr{B}}(\mathscr{B} )&=f_{\mathscr{B} \left| \mathscr{H} _0 \right.}(\mathscr{B} )=f_{\mathscr{B} \left| \mathscr{H} _1 \right.}(\mathscr{B} )\\
			\\&=\frac{\mathscr{B} ^{\mathrm{MM}}\left( 1-\mathscr{B} \right) ^{MN-2}}{\beta \left( \mathrm{MM}+1,MN-1 \right)}.\\
		\end{split}
	\end{equation}
	Therefore, \eqref{eq37} shows that OGLRT own the CFAR property against the covariance matrix under $\mathcal{H}_0$ hypothesis. 
	
	Next, by using the results \cite{Kelly1981FiniteSumEF,kelly1986adaptive}, the expression for 
	the PFA of OGLRT can be written as
	\begin{equation}
		\label{41}
		\mathrm{PFA}=\left( 1+\lambda \right) ^{-\mathrm{MM}}
	\end{equation}
	Finally, we have 
	\begin{equation}
		\label{477}
		P_{\mathrm{D}}^{}=\int_0^1{\left[ 1-\mathcal{P} _{\mathscr{H} _1,\mathscr{B}}\left( \lambda \right) \right]}f_{\mathscr{B}}(\mathscr{B} )d\mathscr{B}, 
	\end{equation}
	with
	\begin{equation}
		\label{48}
		\mathcal{P} _{\mathscr{H} _1,\mathscr{B}}\left( \gamma \right) =\sum_{i=0}^{\mathrm{MM}-1}{C_{\mathrm{MM}}^{1+i}}\frac{\gamma ^{1+i}}{(1+\gamma )^{\mathrm{MM}}}\mathrm{IG}_{i+1}\left( \frac{\alpha \mathscr{B}}{1+\gamma} \right). 
	\end{equation}
	The notation $\mathrm{IG}_{\iota +1}\left( \boldsymbol{\hbar } \right) $ stands for the incomplete gamma function, given by
	\begin{equation}
    \mathrm{IG}_{\iota +1}\left( \boldsymbol{\hbar } \right) =\exp \left( -\boldsymbol{\hbar } \right) \sum_{j=0}^{\iota}{\frac{1}{\varGamma \left( j+1 \right)}}\boldsymbol{\hbar }^j
	\end{equation}

	\subsection{The PFA and PD of TGLRT}   
	In this section, we will derive closed-form expressions for the PFA and PD of the TGLRT detector. For the TGLRT detector in \eqref{eq28}, its detection statistic expression is consistent with the detector in \cite{robey1992cfar}. 
	For the threshold of the TGLRT detector, it is only necessary to rewrite the $\lambda$ in  \eqref{41} as $\lambda\mathscr{B}$. In other words, according to the results in \cite{robey1992cfar}, the relationship between the detection threshold of TGLRT and PFA can be expressed as:
	\begin{equation}
		\label{49}
		\mathrm{PFA}=\int_0^1{\frac{1}{\left( 1+\lambda \mathscr{B} _1 \right) ^{\mathrm{MM}_1}}}f_{\mathscr{B} _1}(\mathscr{B} _1)d\mathscr{B} _1,
	\end{equation}
	where $ \mathrm{MM}_1=LK-MN+1 $ and $\mathscr{B} _1\sim C\beta _{\mathrm{MM}_1,MN-1} $. The PDF of $\mathscr{B}_1$, denoted as $f_{\mathscr{B} _1}(\mathscr{B} _1)$, can be obtained by replacing the degrees of freedom $\mathrm{MM}$ in \eqref{45} with $\mathrm{MM}_1$. Obviously, TGLRT possesses the CFAR property due to the reason that 	
	the PFA expression \eqref{49} is irrelevant with the covariance matrix. 
	
	Additionally, the PD can be expressed as
	\begin{equation}
		\label{500}
		P_{\mathrm{D}}^{}=\int_0^1{\left[ 1-\mathcal{P} _{\mathscr{H} _1,\mathscr{B}_1}\left( \lambda \mathscr{B}_1 \right) \right]}f_{\mathscr{B}_1}(\mathscr{B}_1 )d\mathscr{B}_1. 
	\end{equation}
	After substituting $\mathrm{MM}_1$ for $\mathrm{MM}$ in \eqref{48}, we can obtain the expression for $\mathcal{P} _{\mathscr{H} _1,\mathscr{B} _1}\left( \gamma \right) $.
	
	\subsection{The PFA and PD of Rao}   
	Using \eqref{eq33}, the denominator of \eqref{eq32} can be recast as
	\begin{equation}
		\begin{split}
			\label{46}&
			\tilde{\mathbf{a}}_{}^{\dagger}\left( \mathbf{I}+\tilde{\mathbf{Z}}\tilde{\mathbf{Z}}^{\dagger} \right) ^{-1}\mathbf{a}=\tilde{\mathbf{a}}_{}^{\dagger}\left( \mathbf{I}+\tilde{\mathbf{Z}}\mathbf{P}_{\boldsymbol{\omega }_{\mathrm{D}}^{*}\left( f_{\mathrm{d}} \right)}^{\bot}\tilde{\mathbf{Z}}^{\dagger} \right) ^{-1}\mathbf{a}
			\\&
			\qquad\qquad\qquad \qquad\,\,\,   -\frac{\left| \tilde{\mathbf{a}}_{}^{\dagger}\left( \mathbf{I}+\tilde{\mathbf{Z}}\mathbf{P}_{\boldsymbol{\omega }_{\mathrm{D}}^{*}\left( f_{\mathrm{d}} \right)}^{\bot}\tilde{\mathbf{Z}}^{\dagger} \right) ^{-1}\tilde{\mathbf{Z}}_{\boldsymbol{\omega }} \right|^2}{1+\tilde{\mathbf{Z}}_{\boldsymbol{\omega }}^{\dagger}\left( \mathbf{I}+\tilde{\mathbf{Z}}\mathbf{P}_{\boldsymbol{\omega }_{\mathrm{D}}^{*}\left( f_{\mathrm{d}} \right)}^{\bot}\tilde{\mathbf{Z}}^{\dagger} \right) ^{-1}\tilde{\mathbf{Z}}_{\boldsymbol{\omega }}}.
		\end{split}
	\end{equation}
	Besides, the numerator of \eqref{eq32} can be simplified as
	\begin{equation}
		\begin{split}
			\label{47}&
			\left| \tilde{\mathbf{a}}_{}^{\dagger}\left( \mathbf{I}+\tilde{\mathbf{Z}}\tilde{\mathbf{Z}}^{\dagger} \right) ^{-1}\tilde{\mathbf{Z}}_{\boldsymbol{\omega }} \right|^2
			\\&
			=\frac{\left| \tilde{\mathbf{a}}_{}^{\dagger}\left( \mathbf{I}+\tilde{\mathbf{Z}}\mathbf{P}_{\boldsymbol{\omega }_{\mathrm{D}}^{*}\left( f_{\mathrm{d}} \right)}^{\bot}\tilde{\mathbf{Z}}^{\dagger} \right) ^{-1}\tilde{\mathbf{Z}}_{\boldsymbol{\omega }} \right|^2}{\left| 1+\tilde{\mathbf{Z}}_{\boldsymbol{\omega }}^{\dagger}\left( \mathbf{I}+\tilde{\mathbf{Z}}\mathbf{P}_{\boldsymbol{\omega }_{\mathrm{D}}^{*}\left( f_{\mathrm{d}} \right)}^{\bot}\tilde{\mathbf{Z}}^{\dagger} \right) ^{-1}\tilde{\mathbf{Z}}_{\boldsymbol{\omega }} \right|^2}.
		\end{split}
	\end{equation}
	Applying \eqref{46} and \eqref{47}, after some algebra, Rao detector expression can be rewritten as \eqref{eq467} (which is shown on the first line of next page)
	\begin{figure*}[htp]
		\begin{equation}
			\label{eq467}
			\underset{\text{\underline{\hspace{17cm}}}}{
				\begin{split}&
					\varLambda_{Rao} 
					=\frac{\left| \tilde{\mathbf{a}}_{}^{\dagger}\left( \mathbf{I}+\tilde{\mathbf{Z}}\mathbf{P}_{\boldsymbol{\omega }_{\mathrm{D}}^{*}\left( f_{\mathrm{d}} \right)}^{\bot}\tilde{\mathbf{Z}}^{\dagger} \right) ^{-1}\tilde{\mathbf{Z}}_{\boldsymbol{\omega }} \right|^2}{\left( 1+\tilde{\mathbf{Z}}_{\boldsymbol{\omega }}^{\dagger}\left( \mathbf{I}+\tilde{\mathbf{Z}}\mathbf{P}_{\boldsymbol{\omega }_{\mathrm{D}}^{*}\left( f_{\mathrm{d}} \right)}^{\bot}\tilde{\mathbf{Z}}^{\dagger} \right) ^{-1}\tilde{\mathbf{Z}}_{\boldsymbol{\omega }} \right) \tilde{\mathbf{a}}_{}^{\dagger}\left( \mathbf{I}+\tilde{\mathbf{Z}}\mathbf{P}_{\boldsymbol{\omega }_{\mathrm{D}}^{*}\left( f_{\mathrm{d}} \right)}^{\bot}\tilde{\mathbf{Z}}^{\dagger} \right) ^{-1}\mathbf{a}}
					\\&\qquad
					\times \frac{1}{1+\tilde{\mathbf{Z}}_{\boldsymbol{\omega }}^{\dagger}\left( \mathbf{I}+\tilde{\mathbf{Z}}\mathbf{P}_{\boldsymbol{\omega }_{\mathrm{D}}^{*}\left( f_{\mathrm{d}} \right)}^{\bot}\tilde{\mathbf{Z}}^{\dagger} \right) ^{-1}\tilde{\mathbf{Z}}_{\boldsymbol{\omega }}-\frac{\left| \tilde{\mathbf{a}}_{}^{\dagger}\left( \mathbf{I}+\tilde{\mathbf{Z}}\mathbf{P}_{\boldsymbol{\omega }_{\mathrm{D}}^{*}\left( f_{\mathrm{d}} \right)}^{\bot}\tilde{\mathbf{Z}}^{\dagger} \right) ^{-1}\tilde{\mathbf{Z}}_{\boldsymbol{\omega }} \right|^2}{\tilde{\mathbf{a}}_{}^{\dagger}\left( \mathbf{I}+\tilde{\mathbf{Z}}\mathbf{P}_{\boldsymbol{\omega }_{\mathrm{D}}^{*}\left( f_{\mathrm{d}} \right)}^{\bot}\tilde{\mathbf{Z}}^{\dagger} \right) ^{-1}\mathbf{a}}}	
			\end{split}}
		\end{equation}
	\end{figure*}  
	or equivalently
	\begin{equation}
		\label{544}
		\varLambda_{Rao} =\frac{\varLambda ^{\prime}\varLambda ^{''}}{\varLambda _{LHAMF}}\underset{\mathcal{H} _0}{\overset{\mathcal{H} _1}{\gtrless}}\lambda, 
	\end{equation}
	where the definition of $\varLambda _{LHAMF}$ is 
	\begin{equation}
		\label{50}
		\varLambda _{LHAMF}=\frac{\left| \tilde{\mathbf{a}}^{\dagger}\left( \mathbf{I}+\tilde{\mathbf{Z}}\mathbf{P}_{\boldsymbol{\omega }_{\mathrm{D}}^{*}\left( f_{\mathrm{d}} \right)}^{\bot}\tilde{\mathbf{Z}}^{\dagger} \right) ^{-1}\tilde{\mathbf{Z}}_{\boldsymbol{\omega }}^{} \right|^2}{\tilde{\mathbf{a}}^{\dagger}\left( \mathbf{I}+\tilde{\mathbf{Z}}\mathbf{P}_{\boldsymbol{\omega }_{\mathrm{D}}^{*}\left( f_{\mathrm{d}} \right)}^{\bot}\tilde{\mathbf{Z}}^{\dagger} \right) ^{-1}\tilde{\mathbf{a}}}.
	\end{equation}

It is worth noting that if we define the steering vector, received data, and SCM as $\tilde{\mathbf{a}}$, $
\tilde{\mathbf{Z}}_{\boldsymbol{\omega }}$, $\left( \mathbf{I}+\tilde{\mathbf{Z}}\mathbf{P}_{\boldsymbol{\omega }_{\mathrm{D}}^{*}\left( f_{\mathrm{d}}\right)}^{\bot}\tilde{\mathbf{Z}}^{\dagger} \right) $ respectively, then \eqref{50} is the corresponding TGLRT detector, which is also identical in form to the AMF detector in reference \cite{robey1992cfar}. In order to distinguish the detector represented by \eqref{50} from the previous TGLRT detector in the text, this paper defines it as the LHAMF detector\footnote{ "LH" represents the first letters of the surnames of the first two authors.}.

In contrast to TGLRT, the LHAMF detector utilizes both training and test data for covariance matrix estimation, potentially leading to slightly better performance. Additionally, by changing the degrees of freedom in \eqref{49} and \eqref{500} from $\mathrm {MM}_1$ to $\mathrm {MM}$, the expressions for PFA and PD of the LHAMF detector can be derived, i.e.,
	\begin{equation}
		\label{499}
		\mathrm{PFA}=\int_0^1{\frac{1}{\left( 1+\lambda \mathscr{B}  \right) ^{\mathrm{MM}}}}f_{\mathscr{B} }(\mathscr{B})d\mathscr{B} 
	\end{equation}
	and
	\begin{equation}
		\label{5000}
		P_{\mathrm{D}}^{}=\int_0^1{\left[ 1-\mathcal{P} _{\mathscr{H} _1,\mathscr{B}}\left( \lambda \mathscr{B} \right) \right]}f_{\mathscr{B}}(\mathscr{B} )d\mathscr{B}.
	\end{equation}
	
	Define
	\begin{equation}
		\label{51}
		\mathscr{B} =\frac{\varLambda ^{''}}{\varLambda _{LHAMF}}.
	\end{equation}
	Note that referring to the result in \cite{de2007rao}, the distribution of random variables presented in \eqref{51} aligns with \eqref{39}. Hence, we adopt the same variable symbol. Subsequently, the expression for the Rao detector can be given by
	\begin{equation}
		\varLambda =\frac{\mathscr{B} \varLambda ^{''}}{1+\varLambda ^{''}}\underset{\mathcal{H} _0}{\overset{\mathcal{H} _1}{\gtrless}}\lambda.
	\end{equation}
Similar to \cite{de2007rao}, the PFA for Rao detector can be expressed as
	\begin{equation}
		\label{53}
		\mathrm{PFA}=\left( 1-\lambda \right) ^{\left( L+1 \right) K-1}.
	\end{equation}
	Eq.\eqref{53} implies that there is no covariance matrix term in the expression of PFA for Rao detector. Above result further indicates that Rao detector exhibit CFAR property against noise covariance matrix under $\mathcal{H}_0$ hypothesis.
	
	 Furthermore, the detection probability can be expressed as
	\begin{equation}
		\label{61}
		P_{\mathrm{D}}^{}=\int_{\lambda}^1{\left[ 1-\mathcal{P} _{\mathscr{H} _1,\mathscr{B}}\left( \frac{\lambda}{\mathscr{B} -\lambda} \right) \right] f_{\mathscr{B}}(\mathscr{B} )d\mathscr{B}}
	\end{equation}
	
	Finally, for the convenience of readers, Table I provides the PD and PFA for detectors above mentioned.  
	\begin{table}[htp]
		\label{table1}
		\caption{The statistical expressions, PD and PFA for Proposed Detectors.}
		\vspace{0pt}
		\centering
		\renewcommand{\arraystretch}{1.5} 
		\begin{tabularx}{214pt}{c|cccc}  \hline
			{  }&{OGLRT} &{TGLRT}&{LHAMF}&{Rao}\\	\hline 
			Statistical expression &\eqref{eq200}&\eqref{eq28} &\eqref{50}&\eqref{eq32}\\		
			$\mathrm{PFA}$ &\eqref{41} &\eqref{49} &\eqref{499}&\eqref{53}\\		
			$P_{\mathrm{D}}$&\eqref{477} &\eqref{500}&\eqref{5000}&\eqref{61}\\		\hline  
		\end{tabularx}
	\end{table}
	
	\section{Simulation Results}
	\label{sec5}
	This section will verify the derived theory's correctness through numerical simulation.
	The parameters are set as follows: $M=4$, $N=3$ are the number of transmit and receive elements, respectively. The carrier frequency is $f_0=2$ GHz, and the frequency offset between the elements equals the signal bandwidth, i.e., $B=\Delta f=1$MHz, and Doppler frequency shift is $f_{\mathrm{d}}=0.2$. $d_{\mathrm{T}}=d_{\mathrm{R}}=\frac{c}{2f_0}$ represent the distance between the transmit and receive elements. Further, the Signal-to-Jamming-Plus-Noise Ratio (SJNR) is defined as
	\begin{equation}
		\label{77}
		\mathbf{R}=\sigma _{\mathrm{c}}^{2}\left( \sum_{\mathrm{d}}{\sigma _{\mathrm{jam},\mathrm{d}}^{2}\mathbb{R} _{\mathrm{d}}}+\sum_{\mathrm{u}}{\sigma _{\mathrm{sup},\mathrm{u}}^{2}\mathbb{R} _{\mathrm{u}}}+\mathbf{I} \right), 
	\end{equation} 
	$\sigma _{\mathrm{jam},\mathrm{d}}^{2}$ and $\sigma _{\mathrm{sup},\mathrm{u}}^{2}$ represents the power of the $d $th deception jamming signal and the ${u} $th suppression jamming. $\mathbb{R} _{\mathrm{d}}$ and $ \mathbb{R} _{\mathrm{u}} $  denote the covariance matrix of deceptive jamming and suppressive jamming, respectively. They are defined as
	\begin{equation}
		\begin{split}
			\label{78}&
			\mathbb{R} _{\mathrm{d}}=\mathbf{a}_{\mathrm{TR}}\left( r_{\mathrm{d}},\theta _{\mathrm{d}} \right) \mathbf{a}_{\mathrm{TR}}^{\dagger}\left( r_{\mathrm{d}},\theta _{\mathrm{d}} \right) 
			\\&
			\mathbb{R} _{\mathrm{u}}=\mathbf{1}_{\mathrm{M}\times \mathrm{M}}\otimes \left[ \mathbf{a}_{\mathrm{R}}\left( \theta _{\mathrm{u}} \right) \mathbf{a}_{\mathrm{R}}^{\dagger}\left( \theta _{\mathrm{u}} \right) \right] 
		\end{split}
	\end{equation} 
	where $\mathbf{1}_{\mathrm{M}\times \mathrm{M}}$ is a matrix with all $1$ elements, $\sigma _{\mathrm{c}}^{2}$ and $\mathbf{I}$ represent the power and covariance matrix of noise, respectively. In addition, $\delta _{\mathrm{jam}}=10\log \frac{\sum_{\mathrm{d}}{\sigma _{\mathrm{jam},\mathrm{d}}^{2}}}{\sigma _{c}^{2}}$ and $\delta _{\mathrm{sup}}=10\log\frac{\sum_{\mathrm{u}}{\sigma _{\mathrm{sup},\mathrm{u}}^{2}}}{\sigma _{\mathrm{c}}^{2}}$ represent jamming to noise ratio (JNR) for deception jamming and suppression jamming. At the same time, signal to noise ratio (SNR) is defined as $\mathrm{SNR}=10\log \frac{\left| \xi \right|^2}{\sigma _{\mathrm{n}}^{2}}$. The target is located in $\left( r,\theta \right) =\left( 15.12 \mathrm{Km},30^{\circ} \right) $, two deceptive targets located $\left( r_1,\theta _1 \right) =\left( 15.165 \mathrm{Km},30^{\circ} \right)$, $\left( r_2,\theta _2 \right) =\left( 30.48 \mathrm{Km},28^{\circ} \right)$, and $\delta _\mathrm{jam,1}=\delta _\mathrm{jam,2}=20$ dB, respectively, and one suppressive jammer with JNR of  $\delta _\mathrm{sup}=30$ dB and the angle is $\theta _\mathrm{u}=-20^{\circ}$.
	The detection threshold and PD obtained by $\frac{100}{\mathrm{PFA}}$ and $10^{4}$ independent Monte Carlo (MC) trials, respectively, where $\mathrm{PFA}=10^{-3}$. 

Additionally, for cross-validation purposes, we compared the performance of the statistical expressions with the expressions for PFA and PD.
	It is worth noting that the experimental results obtained from MC and statistical equivalent expressions are marked as "MC" and  "TH", respectively.

	\begin{figure}[htp]
		\centering
		\subfigure[]{\includegraphics[width=0.35\textwidth]{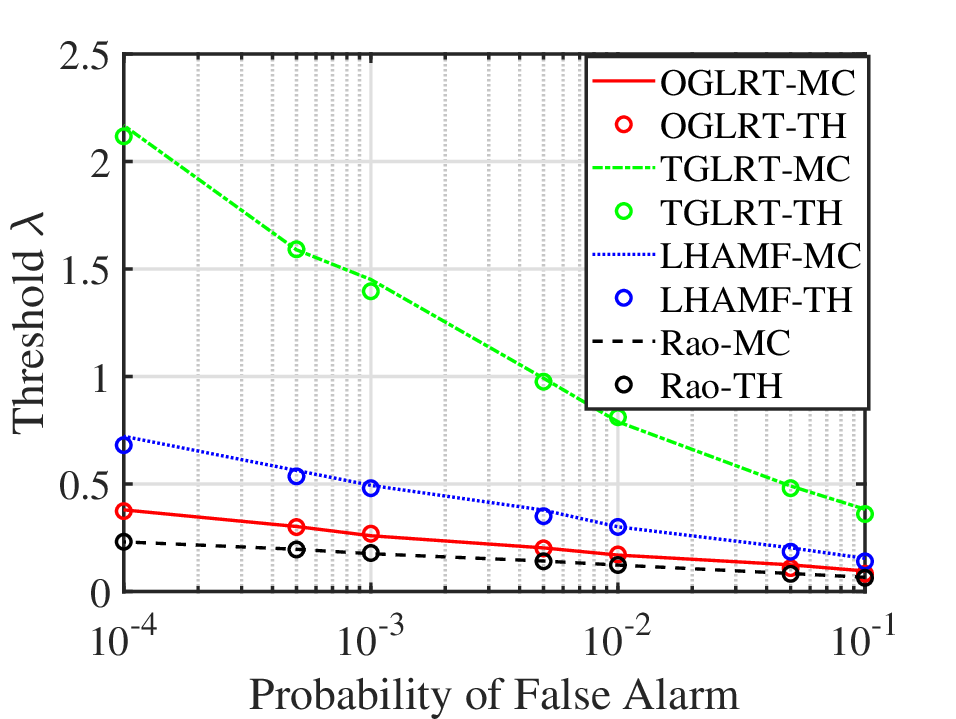}}
		\caption{The detection threshold of the proposed detectors varies with the false alarm probability.}
		\label{fig1}	
	\end{figure}    
Fig.\ref{fig1} shows the variation curves of the detection threshold of the proposed detectors with respect to the false alarm rate $
{\mathrm{PFA}}
$. The lines represent the experimental results obtained through MC simulations, while the symbols $\circ$ represent the theoretical values obtained using the expressions in the second row of Table I. It can be observed from the figure that the theoretical values closely match the results obtained from MC simulations, indicating the correctness of the theoretical derivations in this section.

%
	\begin{figure}[htp]
		\centering
		\subfigure[]{\includegraphics[width=0.35\textwidth]{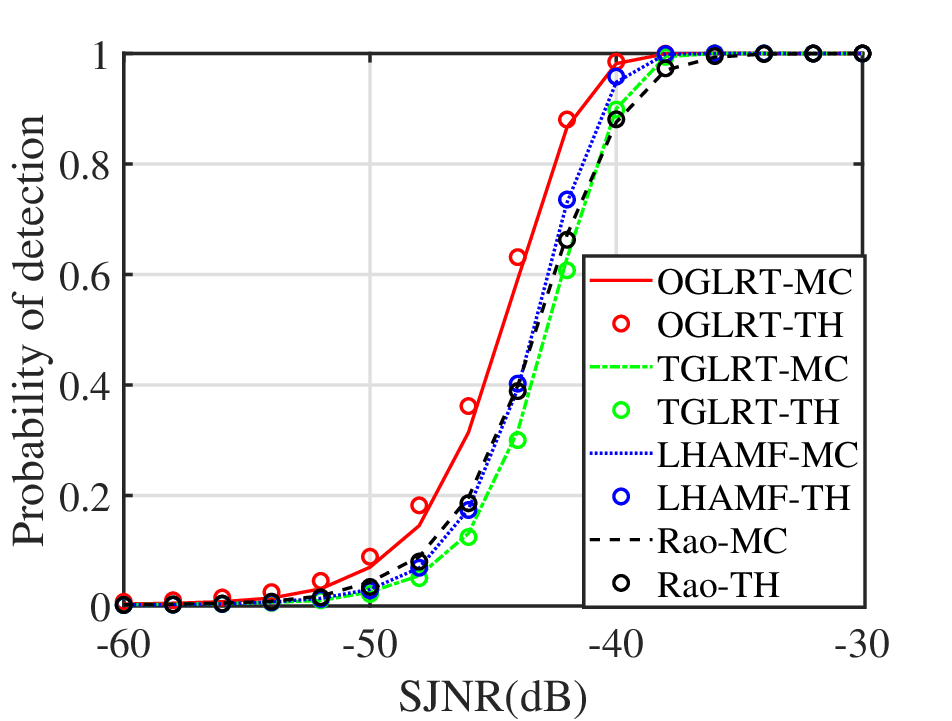}}
		\subfigure[]{\includegraphics[width=0.35\textwidth]{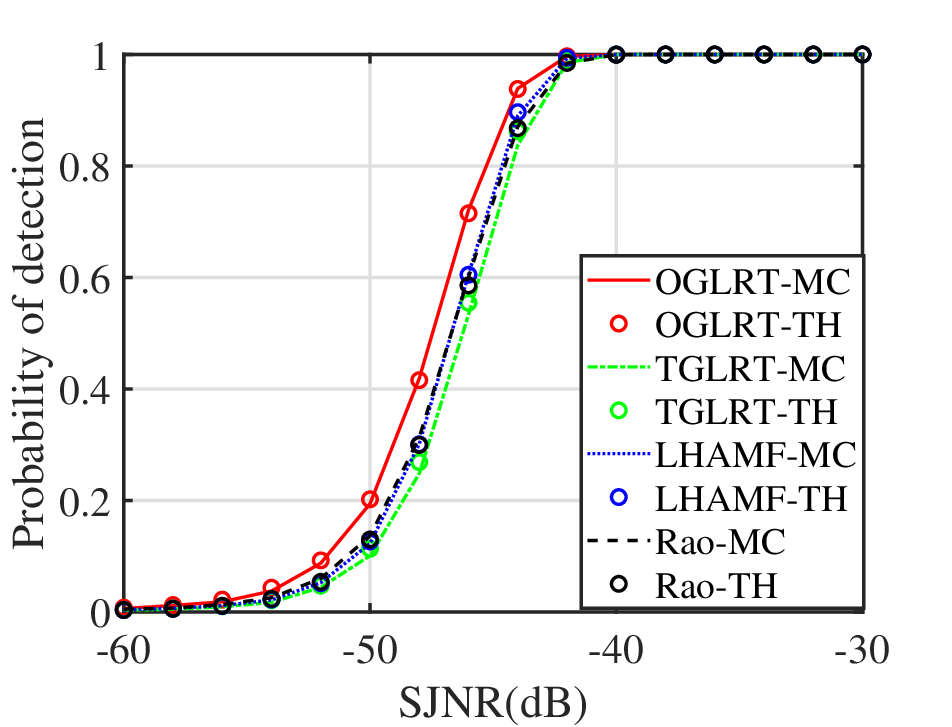}}	
		\subfigure[]{\includegraphics[width=0.35\textwidth]{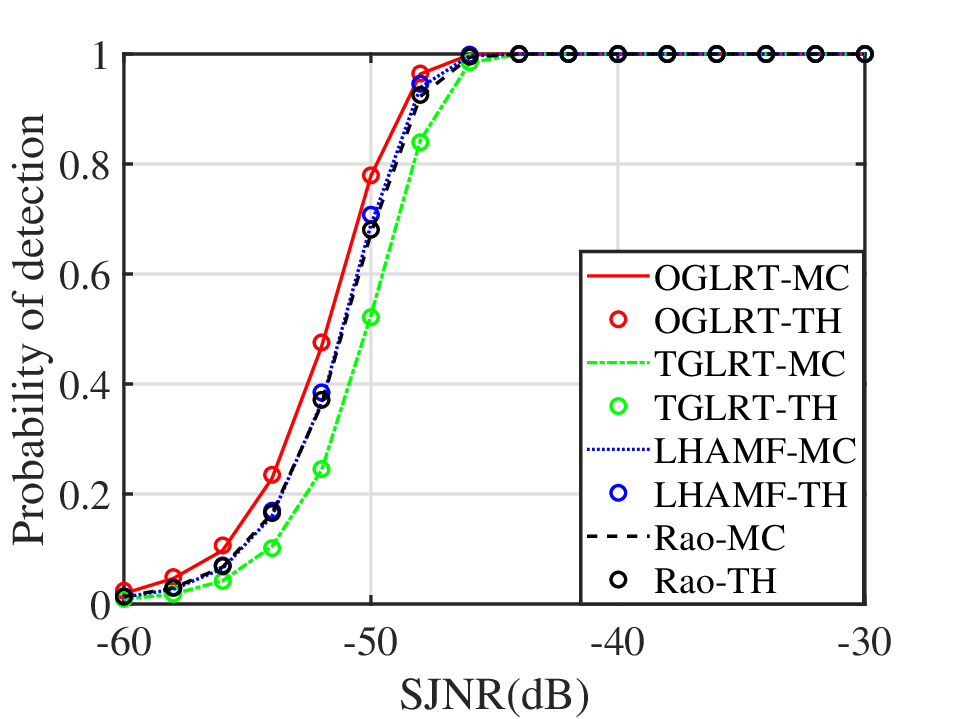}}
		\subfigure[]{\includegraphics[width=0.35\textwidth]{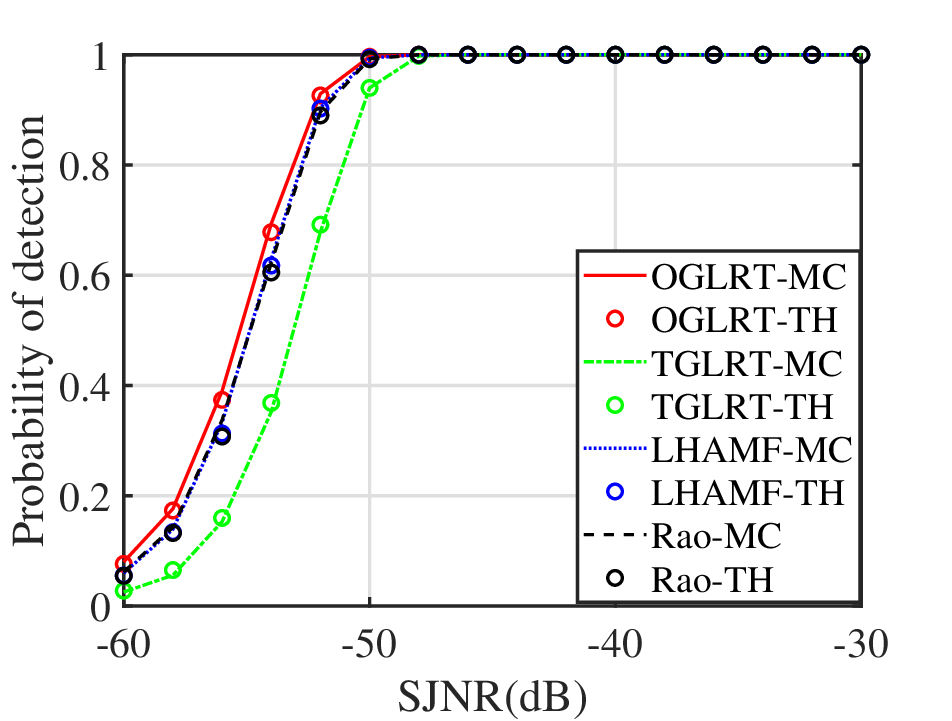}}
		\caption{The curves of PD versus SJNR under different number of training data and snapshots. (a) $L$=4, $K$=6, (b) $L=6$, $K=6$, (c) $L=2$, $K=16$, (d) $L=1$, $K=32$.}
		\label{fig2}
	\end{figure}
	
	Under different numbers of training data and snapshots,  Fig.\ref{fig2} presents the PDs versus SJNR for proposed detectors. The agreement between the MC simulations and theoretical predictions in Fig.\ref{fig2} validates the correctness of the prior theoretical analysis. 
	In Fig.\ref{fig2} (a), where training data is limited, the alignment between MC and theoretical results for OGLRT is weaker compared to Fig.\ref{fig2} (b) with more training data. 
	Additionally, Fig.\ref{fig2} (a) reveals OGLRT's superior performance under the given simulation conditions, while TGLRT performs the poorest, and LHAMF outperforms TGLRT. Comparing Fig.\ref{fig2} (a) and Fig.\ref{fig2} (b), it becomes evident that having more training data enhances detector performance under the same number of snapshots, allowing for more accurate covariance matrix estimation.
	Comparing Fig.\ref{fig2} (b) and Fig.\ref{fig2} (c) suggests that, with limited training data, increasing the number of snapshots also improves the detection performance. 
	In both cases, the data used for covariance matrix estimation already exceeds $2MN$. Thus, enhancing the accuracy of covariance matrix estimation through increased training data yields marginal improvements. However, in this scenario, increasing the number of snapshots significantly enhances detection performance by improving covariance estimation accuracy, as demonstrated in Fig.\ref{fig2} (d).
	\begin{figure}[htp]
		\centering
		\subfigure[]{\includegraphics[width=0.35\textwidth]{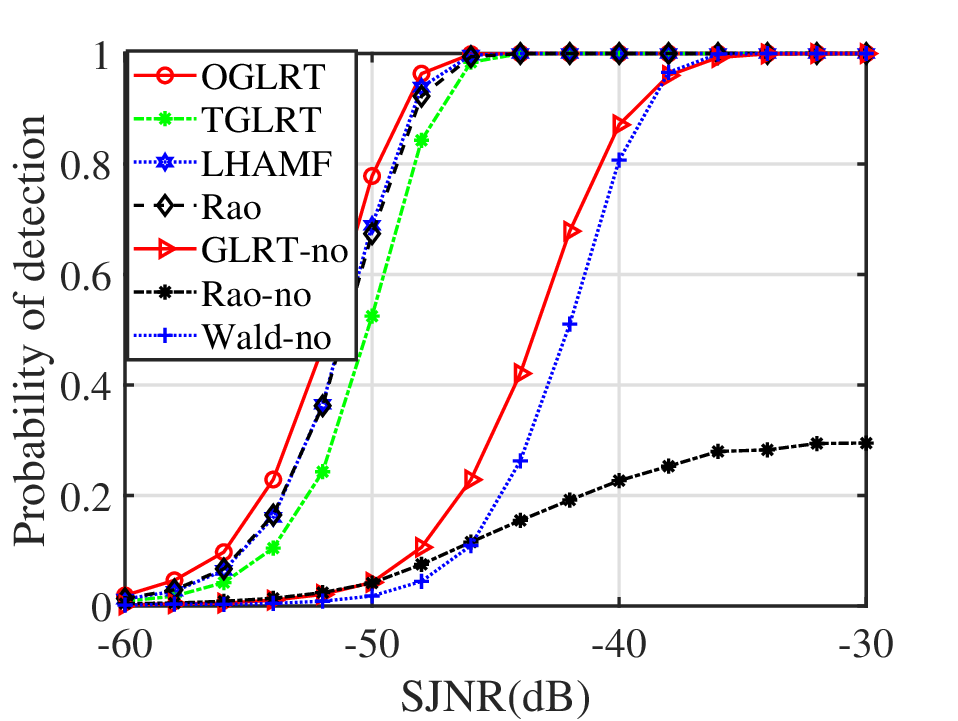}}
		\subfigure[]{\includegraphics[width=0.35\textwidth]{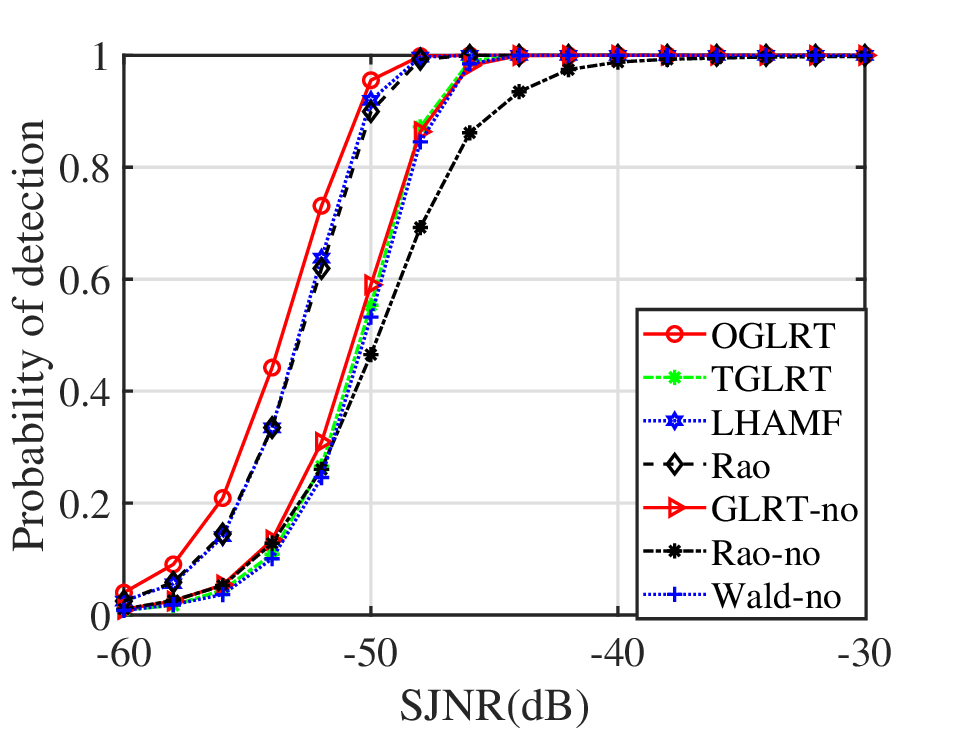}}	
		\caption{The curves of PD with and without training data versus SJNR  (a) $L=2$, $K=16$, (b) $L=1$, $K=24$.}
		\label{fig3}
	\end{figure}
	
To further demonstrate the impact of training data on detector performance, three detectors without training data, namely GLRT-no, Rao-no, and Wald-no, are introduced as comparative detectors, expressed as \cite{huang2021adaptive,GUI2020102861}	
 	\begin{equation}
    \begin{split}&
\varLambda _{GLRT-no}=\frac{\mathbf{a}_{\mathrm{TR}}^{\dagger}\left( r,\theta \right) \left( \mathbf{ZP}_{\boldsymbol{\omega }_{\mathrm{D}}^{*}\left( f_{\mathrm{d}} \right)}^{\bot}\mathbf{Z}^{\dagger} \right) ^{-1}\mathbf{a}_{\mathrm{TR}}^{}\left( r,\theta \right)}{\mathbf{a}_{\mathrm{TR}}^{\dagger}\left( r,\theta \right) \left( \tilde{\mathbf{Z}}\tilde{\mathbf{Z}}^{\dagger} \right) ^{-1}\mathbf{a}_{\mathrm{TR}}^{}\left( r,\theta \right)}\underset{\mathcal{H} _0}{\overset{\mathcal{H} _1}{\gtrless}}\lambda, 
	\end{split}
    \end{equation} 	
	\begin{equation}
	\begin{split}&
\varLambda _{Rao-no}=\frac{\left| \mathbf{a}_{\mathrm{TR}}^{\dagger}\left( r,\theta \right) \left( \mathbf{ZZ}^{\dagger} \right) ^{-1}\mathbf{Z}\boldsymbol{\omega }_{\mathrm{D}}^{*}\left( f_{\mathrm{d}} \right) \right|^2}{\mathbf{a}_{\mathrm{TR}}^{\dagger}\left( r,\theta \right) \left( \mathbf{ZZ}^{\dagger} \right) ^{-1}\mathbf{a}_{\mathrm{TR}}^{}\left( r,\theta \right) \boldsymbol{\omega }_{\mathrm{D}}^{T}\left( f_{\mathrm{d}} \right) \boldsymbol{\omega }_{\mathrm{D}}^{*}\left( f_{\mathrm{d}} \right)}
\\&
\qquad\quad\,\,\,  \underset{\mathcal{H} _0}{\overset{\mathcal{H} _1}{\gtrless}}\lambda, 
	\end{split}
	\end{equation} 
	\begin{equation}
\begin{split}&
\varLambda _{Wald-no}=\frac{\left| \mathbf{a}_{\mathrm{TR}}^{\dagger}\left( r,\theta \right) \left( \mathbf{ZP}_{\boldsymbol{\omega }_{\mathrm{D}}^{*}\left( f_{\mathrm{d}} \right)}^{\bot}\mathbf{Z}^{\dagger} \right) ^{-1}\mathbf{Z}\boldsymbol{\omega }_{\mathrm{D}}^{*}\left( f_{\mathrm{d}} \right) \right|^2}{\mathbf{a}_{\mathrm{TR}}^{\dagger}\left( r,\theta \right) \left( \mathbf{ZP}_{\boldsymbol{\omega }_{\mathrm{D}}^{*}\left( f_{\mathrm{d}} \right)}^{\bot}\mathbf{Z}^{\dagger} \right) ^{-1}\mathbf{a}_{\mathrm{TR}}^{}\left( r,\theta \right) \boldsymbol{\omega }_{\mathrm{D}}^{T}\left( f_{\mathrm{d}} \right) \boldsymbol{\omega }_{\mathrm{D}}^{*}\left( f_{\mathrm{d}} \right)}
\\&
\qquad\qquad \underset{\mathcal{H} _0}{\overset{\mathcal{H} _1}{\gtrless}}\lambda. 
	\end{split}
    \end{equation} 
    
	Fig. \ref{fig3} compares the performance of the proposed detectors with those lacking training data with $L=2$, $K=16$, and $L=1$, $K=24$.
	In Fig. \ref{fig3} (a), it is evident that the performance of the Rao detector without training data is the poorest with fewer snapshots. This result is expected, given that the smaller value of $K$ hampers the acquisition of an accurate covariance matrix.
	Contrastingly, with the support of training data, the proposed detectors demonstrate the ability to detect targets even with a limited number of snapshots, showcasing notably superior performance compared to detectors without training data. In Fig. \ref{fig3} (b), it is demonstrated that even with a more significant number of snapshots, the performance of the three detectors without training data remains inferior to the proposed detectors, which benefit from training data.
	
	\begin{figure}[htp]
		\centering
		\includegraphics[width=0.35\textwidth]{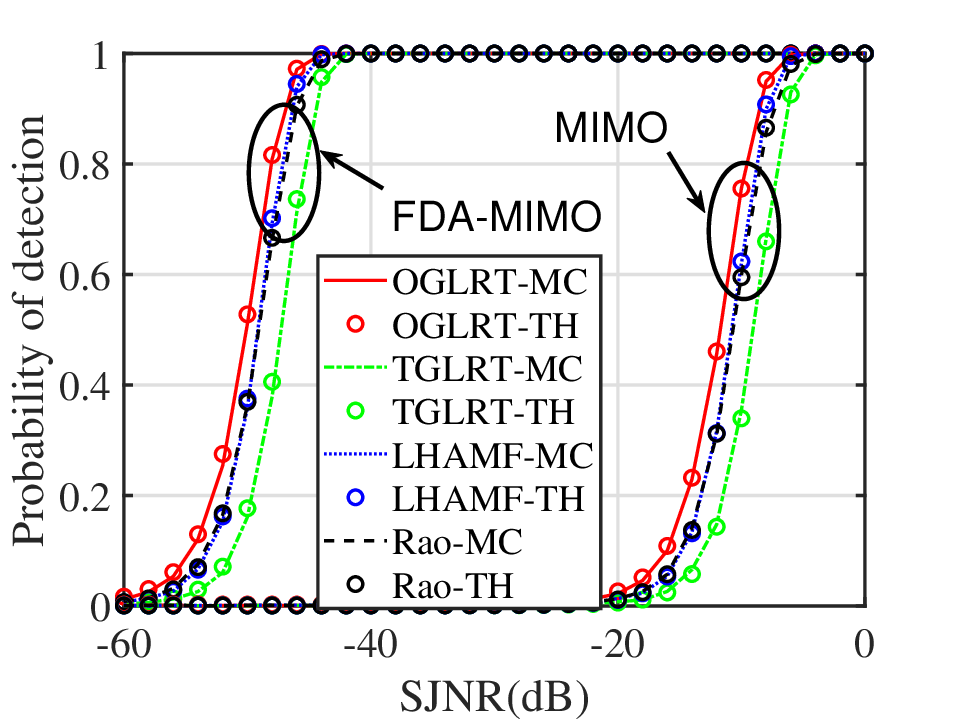}
		\caption{Comparison of detection performance of the proposed detectors for FDA-MIMO and MIMO radars, where $L=2$, $K=12$.}
		\label{fig4}	
	\end{figure}
	Furthermore, by setting the frequency offset $\bigtriangleup f=0$ in the transmitting-receiving steering vector of FDA-MIMO radar, we obtain the transmitting-receiving steering vector for MIMO radar. For the detection problem discussed in this paper, Fig. \ref{fig4} presents the PDs versus SJNR for FDA-MIMO and MIMO radar when $L=2$ and $K=12$. The results in Fig. \ref{fig4} indicate that the detectors proposed in this paper and the corresponding statistical equivalent theory also apply to MIMO radar. In other words, the MC results for MIMO radar align well with the theoretical predictions.
	Specifically, OGLRT exhibits the best performance, LHAMF and Rao detectors perform nearly identically, and TGLRT demonstrates the poorest performance. Additionally, under the given parameters, the PD of FDA-MIMO radar significantly outperforms that of MIMO radar.
	
	\begin{figure}[htp]\centering
		\subfigure[]
		{\includegraphics[width=0.35\textwidth]{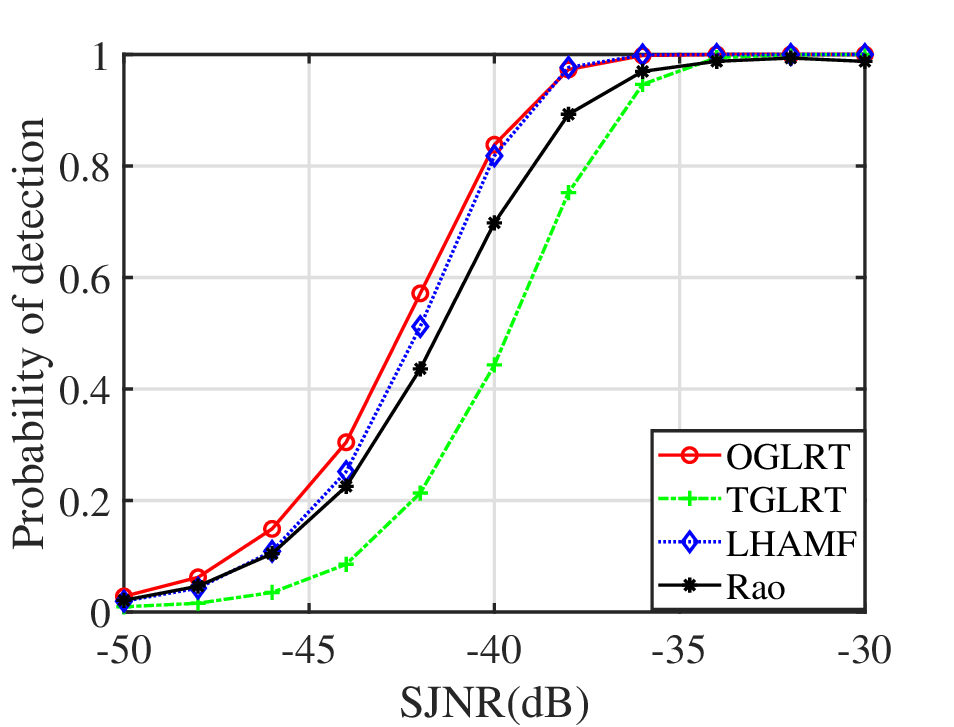}}
		\subfigure[]
		{\includegraphics[width=0.35\textwidth]{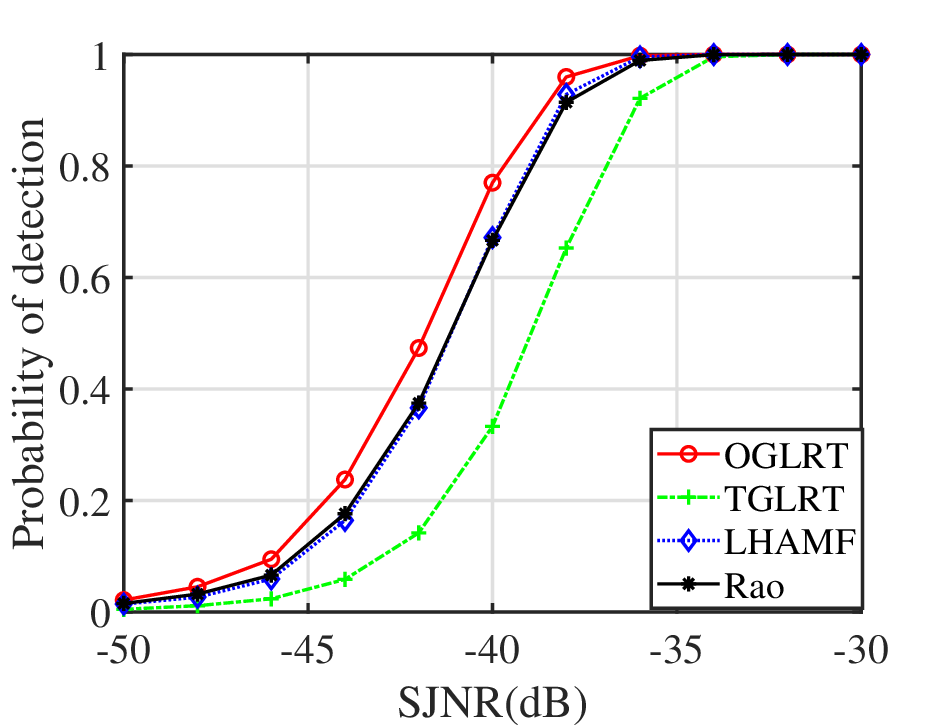}}
		\subfigure[]{
			\includegraphics[width=0.35\textwidth]{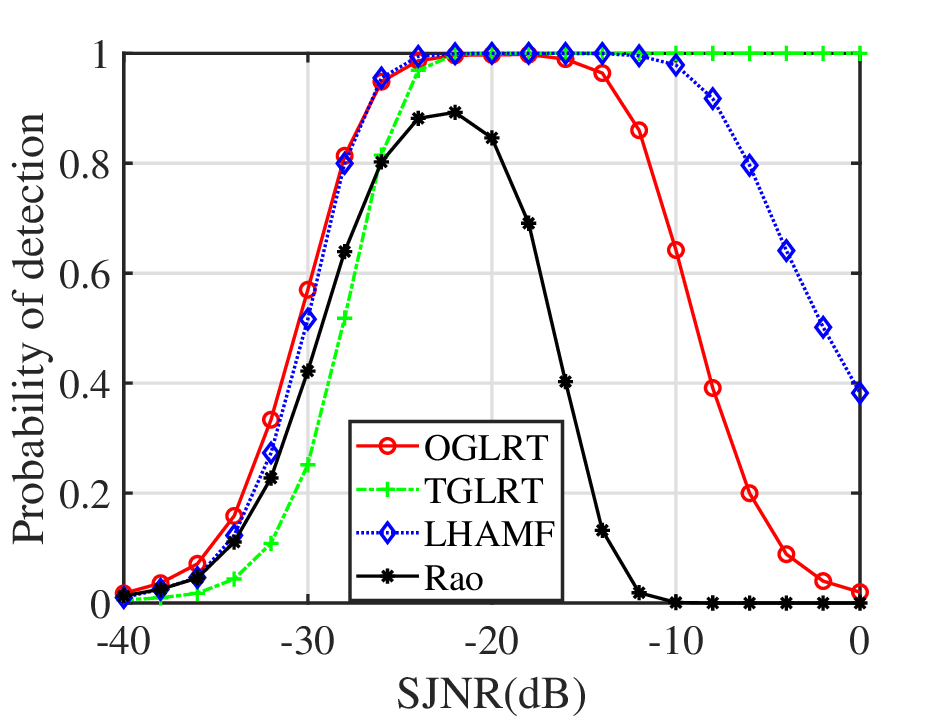}}
		\caption{PD versus SJNR when the signal mismatch, where $L=1$, $K=24$. (a) $\cos ^2\varphi=0.76$, (b) $\cos ^2\varPhi=0.76$, (c) $\cos ^2\varphi=0.76$, $\cos ^2\varPhi=0.76$.}
		\label{fig5}
	\end{figure}
	
	However, in practice, due to the errors in array calibration, waveform mismatch, and so on, the actual signal steering vector does not fully match the nominal one \cite{gu2012robust,sun2021robust,bandiera2022advanced}. To measure the error, we define $\varphi$ represent the angle between actual and nominal steering vector, the degree of mismatch can be written as\cite{bose1996adaptive,liu2014adaptive,liu2019robust,Liu2014AdaptiveHE}
	\begin{equation}
		\begin{split}
			\label{78}&
			\cos ^2\varphi =\frac{\mathbf{a}_{\mathrm{TR}}^{\dagger}\left( r,\theta \right) \mathbf{R}^{-1}\mathbf{a}_{\mathrm{TR}}\left( r_0,\theta _0 \right)}{\mathbf{a}_{\mathrm{TR}}^{\dagger}\left( r,\theta \right) \mathbf{R}^{-1}\mathbf{a}_{\mathrm{TR}}\left( r,\theta \right)}
			\\&
			\qquad \qquad\,\,         \times \frac{\mathbf{a}_{\mathrm{TR}}^{\dagger}\left( r_0,\theta _0 \right) \mathbf{R}^{-1}\mathbf{a}_{\mathrm{TR}}\left( r,\theta \right)}{\mathbf{a}_{\mathrm{TR}}^{\dagger}\left( r_0,\theta _0 \right) \mathbf{R}^{-1}\mathbf{a}_{\mathrm{TR}}\left( r_0,\theta _0 \right)},
		\end{split}
	\end{equation}
	where, $\mathbf{a}_{\mathrm{TR}}^{\dagger}\left( r,\theta \right)$ and $\mathbf{a}_{\mathrm{TR}}\left( r_0,\theta _0 \right)$
	represent the actual and nominal transmitting-receiving steering vector, respectively. Similarly, to measure the signal mismatch in the Doppler dimension, we have
	\begin{equation}
		\label{eq3-74}
		\cos ^2\varPhi =\frac{\left| \boldsymbol{\omega }_{\mathrm{D}}^{\dagger}\left( f_{\mathrm{d},0} \right) \boldsymbol{\omega }_{\mathrm{D}}^{}\left( f_{\mathrm{d}} \right) \right|^2}{\left\| \boldsymbol{\omega }_{\mathrm{D}}^{\dagger}\left( f_{\mathrm{d},0} \right) \right\| ^2\left\| \boldsymbol{\omega }_{\mathrm{D}}^{}\left( f_{\mathrm{d}} \right) \right\| ^2},
	\end{equation}
	where $\boldsymbol{\omega }_{\mathrm{D}}^{}\left( f_{\mathrm{d},0}\right)$ and $\boldsymbol{\omega }_{\mathrm{D}}^{}\left( f_{\mathrm{d}}\right)$ are the actual and nominal Doppler steering vector, respectively.
	
	Regarding mismatched signals with $L=1$ and $K=24$, Fig.\ref{fig5} illustrates the PDs of the proposed detectors against SJNR. In Fig.\ref{fig5} (a), when there is no error in the Doppler dimension and only mistakes in the transmitting-receiving steering vector that cause signal mismatch, it is evident that signal mismatch significantly degrades detection performance. OGLRT exhibits the best performance in this scenario, while TGLRT's performance remains poor. Furthermore, Fig.\ref{fig5} (b) considers only signal mismatch in the Doppler dimension\footnote{In deriving the statistical expression, the signal mismatch was not considered. This is because no available literature incorporates signal mismatch in the Doppler dimension into consideration, and the authors believe this is a topic worth investigating in future research. Therefore, the results presented in Fig.\ref{fig5}  are all obtained through MC simulations.}
	, resulting in more significant performance degradation compared to Fig.\ref{fig5} (a).
	Finally, in Fig.\ref{fig5} (c), considering signal mismatch in both the transmitting-receiving steering vector and Doppler dimensions, the PDs for three detectors, except TGLRT, are not a single-valued function of SJNR. At low SJNR, the performance of LHAMF and OGLRT is similar and slightly better than the other two detectors. In contrast, at high SJNR, TGLRT exhibits the best detection performance, followed by LHAMF, while the performance of the Rao detector is poor and may not even work.

	In summary, under no signal mismatch or only slight mismatch, the OGLRT detector is recommended, and TGLRT should be the last choice. However, in the presence of significant signal errors, TGLRT emerges as the most robust choice.
	
	\section{Conclusion}
	\label{sec6}
This study delves into the detection of a moving target in Gaussian noise with an unknown covariance matrix using FDA-MIMO radar, introducing four adaptive detectors: OGLRT, TGLRT, Rao, and LHAMF. Closed-form expressions for the PFA and PD, along with confirmation of CFAR properties, are provided.
Simulation results highlight the improved performance obtained by leveraging training data for covariance matrix estimation. The mismatch in the Doppler dimension has a more pronounced impact on detector performance compared to the mismatch in the transmitting-receiving steering vector. We recommend selecting the optimal detector based on the level of signal mismatch. OGLRT is preferable in scenarios with signal matching or slight mismatch, while TGLRT exhibits robust performance in cases of significant signal mismatch. FDA-MIMO radar demonstrates significant advantages over traditional MIMO in Gaussian noise environments with deceptive mainlobe jamming. The closed-form expressions for PFA and PD in the presence of Doppler signal mismatch can serve as a direction for future research.
	
	\appendix
	\section{OGLRT}
	\label{appA}
	\subsection{OGLRT}
	Resorting to one-step GLRT \cite{kelly1986adaptive} criterion, the OGLRT detector can be given by
	\begin{equation}
		\label{eq7}
		\varLambda =\frac{\underset{\xi ,\mathbf{R}}{\max}f\left( \mathbf{Z},\mathbf{Y}\left| \xi \right. ,\mathbf{R},\mathcal{H} _1 \right)}{\underset{\mathbf{R}}{\max}f\left( \mathbf{Z},\mathbf{Y}\left| \mathbf{R},\mathcal{H} _0 \right. \right)}\underset{\mathcal{H} _0}{\overset{\mathcal{H} _1}{\gtrless}}\lambda.
	\end{equation}
	According to \eqref{eq3}, the MLE of $\mathbf{R}$ under $\mathcal{H}_1$ is
	\begin{equation}
		\label{eq8}
		\hat{\mathbf{R}}_1=\frac{\mathcal{Z} _{1}^{}\mathcal{Z} _{1}^{\dagger}+\mathbf{S}}{K\left( L+1 \right)}.
	\end{equation}
	Define
	\begin{align}
		\label{eq11}
		\tilde{\mathbf{Z}}=&\mathbf{S}^{-1/2}\mathbf{Z},\\
		\label{eq12}
		\tilde{\mathbf{a}}=&\mathbf{S}^{-1/2}\mathbf{a}_{\mathrm{TR}}\left( r,\theta \right) .
	\end{align}
	Inserting \eqref{eq8} into \eqref{eq3}, result in
	\begin{equation}
		\label{eq9}
		\begin{split}&
			f\left( \mathbf{Z},\mathbf{Y}\left| \xi \right. ,\hat{\mathbf{R}}_1,\mathcal{H} _1 \right) =\frac{\mathrm{cont}.}{\det ^{K\left( L+1 \right)}\left( \mathcal{Z} _{1}^{}\mathcal{Z} _{1}^{\dagger}+\mathbf{S} \right)}
			\\&
			\qquad\qquad\qquad\qquad\,\,\,\,  
			=\frac{\mathrm{cont}.}{\det ^{K\left( L+1 \right)}\left( \mathbf{S} \right) \left[ \mathcal{G} \left( \xi \right) \right] ^{K\left( L+1 \right)}}.
		\end{split}
	\end{equation}
	where
	\begin{equation}
		\label{eq10}
		\begin{split}&
			\mathcal{G} \left( \xi \right) =\det \left[ \mathbf{I}+\left( \tilde{\mathbf{Z}}-\xi \tilde{\mathbf{a}}\boldsymbol{\omega }_{\mathrm{D}}^{T}\left( f_{\mathrm{d}} \right) \right) \left( \tilde{\mathbf{Z}}-\xi \tilde{\mathbf{a}}\boldsymbol{\omega }_{\mathrm{D}}^{T}\left( f_{\mathrm{d}} \right) \right) ^{\dagger} \right] 
			\\&
			\qquad=\det \left[ \mathbf{I}+\tilde{\mathbf{Z}}\mathbf{P}_{\boldsymbol{\omega }_{\mathrm{D}}^{*}\left( f_{\mathrm{d}} \right)}^{\bot}\tilde{\mathbf{Z}}^{\dagger} \right] \det \left[ \mathbf{I}+\boldsymbol{\omega }_{\mathrm{D}}^{T}\left( f_{\mathrm{d}} \right) \boldsymbol{\omega }_{\mathrm{D}}^{*}\left( f_{\mathrm{d}} \right) \right. 
			\\&
			\qquad\,\,\,\left. \times \left( \xi \tilde{\mathbf{a}}-\frac{\tilde{\mathbf{Z}}\boldsymbol{\omega }_{\mathrm{D}}^{*}\left( f_{\mathrm{d}} \right)}{\boldsymbol{\omega }_{\mathrm{D}}^{T}\left( f_{\mathrm{d}} \right) \boldsymbol{\omega }_{\mathrm{D}}^{*}\left( f_{\mathrm{d}} \right)} \right) \left( \xi \tilde{\mathbf{a}}-\frac{\tilde{\mathbf{Z}}\boldsymbol{\omega }_{\mathrm{D}}^{*}\left( f_{\mathrm{d}} \right)}{\boldsymbol{\omega }_{\mathrm{D}}^{T}\left( f_{\mathrm{d}} \right) \boldsymbol{\omega }_{\mathrm{D}}^{*}\left( f_{\mathrm{d}} \right)} \right) ^{\dagger} \right] 
			\\&
			\qquad=\det \left[ \mathbf{I}+\tilde{\mathbf{Z}}\mathbf{P}_{\boldsymbol{\omega }_{\mathrm{D}}^{*}\left( f_{\mathrm{d}} \right)}^{\bot}\tilde{\mathbf{Z}}^{\dagger} \right] h\left( \xi \right) ,
		\end{split}
	\end{equation}
	
	and
	\begin{equation}
		\label{eq13}
		\begin{split}&
			h\left( \xi \right) =1+\boldsymbol{\omega }_{\mathrm{D}}^{T}\left( f_{\mathrm{d}} \right) \boldsymbol{\omega }_{\mathrm{D}}^{*}\left( f_{\mathrm{d}} \right) \left( \xi \tilde{\mathbf{a}}-\frac{\tilde{\mathbf{Z}}\boldsymbol{\omega }_{\mathrm{D}}^{*}\left( f_{\mathrm{d}} \right)}{\boldsymbol{\omega }_{\mathrm{D}}^{T}\left( f_{\mathrm{d}} \right) \boldsymbol{\omega }_{\mathrm{D}}^{*}\left( f_{\mathrm{d}} \right)} \right) ^{\dagger}
			\\&
			\qquad\,\,\,\,\,\times \left( \mathbf{I}+\tilde{\mathbf{Z}}\mathbf{P}_{\boldsymbol{\omega }_{\mathrm{D}}^{*}\left( f_{\mathrm{d}} \right)}^{\bot}\tilde{\mathbf{Z}}^{\dagger} \right) ^{-1}\left( \xi \tilde{\mathbf{a}}-\frac{\tilde{\mathbf{Z}}\boldsymbol{\omega }_{\mathrm{D}}^{*}\left( f_{\mathrm{d}} \right)}{\boldsymbol{\omega }_{\mathrm{D}}^{T}\left( f_{\mathrm{d}} \right) \boldsymbol{\omega }_{\mathrm{D}}^{*}\left( f_{\mathrm{d}} \right)} \right).
		\end{split}
	\end{equation}
	To find the MLE of $\xi$ in \eqref{eq9}, taking \eqref{eq13} w.r.t. $\xi$, yields
	\begin{equation}
		\begin{split}
			\label{eq14}&
			\frac{\partial h\left( \xi \right)}{\partial \xi}=\boldsymbol{\omega }_{\mathrm{D}}^{T}\left( f_{\mathrm{d}} \right) \boldsymbol{\omega }_{\mathrm{D}}^{*}\left( f_{\mathrm{d}} \right) \left( \xi \mathbf{\tilde{a}}-\frac{\mathbf{\tilde{Z}}\boldsymbol{\omega }_{\mathrm{D}}^{*}\left( f_{\mathrm{d}} \right)}{\boldsymbol{\omega }_{\mathrm{D}}^{T}\left( f_{\mathrm{d}} \right) \boldsymbol{\omega }_{\mathrm{D}}^{*}\left( f_{\mathrm{d}} \right)} \right) ^{\dagger}
			\\&
			\qquad \qquad \,\,\times \left( \mathbf{I}+\mathbf{\tilde{Z}P}_{\boldsymbol{\omega }_{\mathrm{D}}^{*}\left( f_{\mathrm{d}} \right)}^{\bot}\mathbf{\tilde{Z}}^{\dagger} \right) ^{-1}\mathbf{\tilde{a}}.
		\end{split}
	\end{equation}
	Setting the result of \eqref{eq14} to zero, we can obtain the MLE of $\xi$ under $\mathcal{H}_1$, as
	\begin{equation}
		\label{eq15}
		\hat{\xi}=\frac{\mathbf{\tilde{a}}^{\dagger}\left( \mathbf{I}+\mathbf{\tilde{Z}P}_{\boldsymbol{\omega }_{\mathrm{D}}^{*}\left( f_{\mathrm{d}} \right)}^{\bot}\mathbf{\tilde{Z}}^{\dagger} \right) ^{-1}\mathbf{\tilde{Z}}\boldsymbol{\omega }_{\mathrm{D}}^{*}\left( f_{\mathrm{d}} \right)}{\mathbf{\tilde{a}}^{\dagger}\left( \mathbf{I}+\mathbf{\tilde{Z}P}_{\boldsymbol{\omega }_{\mathrm{D}}^{*}\left( f_{\mathrm{d}} \right)}^{\bot}\mathbf{\tilde{Z}}^{\dagger} \right) ^{-1}\mathbf{\tilde{a}}\boldsymbol{\omega }_{\mathrm{D}}^{T}\left( f_{\mathrm{d}} \right) \boldsymbol{\omega }_{\mathrm{D}}^{*}\left( f_{\mathrm{d}} \right)}.
	\end{equation}
	Sustituting \eqref{eq15} into \eqref{eq13}, leads to
	\begin{equation}
		\label{eq16}
		\begin{split}&
			h\left( \hat{\xi} \right) =1-\frac{\left| \tilde{\mathbf{a}}^{\dagger}\left( \mathbf{I}+\tilde{\mathbf{Z}}\mathbf{P}_{\boldsymbol{\omega }_{\mathrm{D}}^{*}\left( f_{\mathrm{d}} \right)}^{\bot}\tilde{\mathbf{Z}}^{\dagger} \right) ^{-1}\tilde{\mathbf{Z}}\boldsymbol{\omega }_{\mathrm{D}}^{*}\left( f_{\mathrm{d}} \right) \right|^2}{\tilde{\mathbf{a}}^{\dagger}\left( \mathbf{I}+\tilde{\mathbf{Z}}\mathbf{P}_{\boldsymbol{\omega }_{\mathrm{D}}^{*}\left( f_{\mathrm{d}} \right)}^{\bot}\tilde{\mathbf{Z}}^{\dagger} \right) ^{-1}\tilde{\mathbf{a}}\boldsymbol{\omega }_{\mathrm{D}}^{T}\left( f_{\mathrm{d}} \right) \boldsymbol{\omega }_{\mathrm{D}}^{*}\left( f_{\mathrm{d}} \right)}
			\\&
			\qquad\qquad\,\,  +\frac{\boldsymbol{\omega }_{\mathrm{D}}^{T}\left( f_{\mathrm{d}} \right) \tilde{\mathbf{Z}}^{\dagger}\left( \mathbf{I}+\tilde{\mathbf{Z}}\mathbf{P}_{\boldsymbol{\omega }_{\mathrm{D}}^{*}\left( f_{\mathrm{d}} \right)}^{\bot}\tilde{\mathbf{Z}}^{\dagger} \right) ^{-1}\tilde{\mathbf{Z}}\boldsymbol{\omega }_{\mathrm{D}}^{*}\left( f_{\mathrm{d}} \right)}{\boldsymbol{\omega }_{\mathrm{D}}^{T}\left( f_{\mathrm{d}} \right) \boldsymbol{\omega }_{\mathrm{D}}^{*}\left( f_{\mathrm{d}} \right)}
			\\&
			\qquad \,\,    =1+\frac{\boldsymbol{\omega }_{\mathrm{D}}^{T}\left( f_{\mathrm{d}} \right) \tilde{\mathbf{Z}}^{\dagger}\left( \mathbf{I}+\tilde{\mathbf{Z}}\mathbf{P}_{\boldsymbol{\omega }_{\mathrm{D}}^{*}\left( f_{\mathrm{d}} \right)}^{\bot}\tilde{\mathbf{Z}}^{\dagger} \right) ^{-1}}{\tilde{\mathbf{a}}^{\dagger}\left( \mathbf{I}+\tilde{\mathbf{Z}}\mathbf{P}_{\boldsymbol{\omega }_{\mathrm{D}}^{*}\left( f_{\mathrm{d}} \right)}^{\bot}\tilde{\mathbf{Z}}^{\dagger} \right) ^{-1}\tilde{\mathbf{a}}}
			\\&
			\qquad\qquad \times \left[ \tilde{\mathbf{a}}^{\dagger}\left( \mathbf{I}+\tilde{\mathbf{Z}}\mathbf{P}_{\boldsymbol{\omega }_{\mathrm{D}}^{*}\left( f_{\mathrm{d}} \right)}^{\bot}\tilde{\mathbf{Z}}^{\dagger} \right) ^{-1}\tilde{\mathbf{a}}\mathbf{I} \right. 
			\\&
			\qquad \qquad \qquad -\left. \tilde{\mathbf{a}}\tilde{\mathbf{a}}^{\dagger}\left( \mathbf{I}+\tilde{\mathbf{Z}}\mathbf{P}_{\boldsymbol{\omega }_{\mathrm{D}}^{*}\left( f_{\mathrm{d}} \right)}^{\bot}\tilde{\mathbf{Z}}^{\dagger} \right) ^{-1} \right] 
			\\&
			\qquad\qquad \times \frac{\tilde{\mathbf{Z}}\boldsymbol{\omega }_{\mathrm{D}}^{*}\left( f_{\mathrm{d}} \right)}{\boldsymbol{\omega }_{\mathrm{D}}^{T}\left( f_{\mathrm{d}} \right) \boldsymbol{\omega }_{\mathrm{D}}^{*}\left( f_{\mathrm{d}} \right)}.
		\end{split}
	\end{equation}
	Further, \eqref{eq10} can be written as  \eqref{eq270}, which is located on the first line of the next page.
	\begin{figure*}[htp]
		\begin{equation}
			\label{eq270}
			\underset{\text{\underline{\hspace{18cm}}}}{
				\begin{split}&
					\mathcal{G} \left( \hat{\xi} \right) =\det \left[ \mathbf{I}+\tilde{\mathbf{Z}}\tilde{\mathbf{Z}}^{\dagger}-\frac{\tilde{\mathbf{a}}\tilde{\mathbf{a}}^{\dagger}\left( \mathbf{I}+\tilde{\mathbf{Z}}\mathbf{P}_{\boldsymbol{\omega }_{\mathrm{D}}^{*}\left( f_{\mathrm{d}} \right)}^{\bot}\tilde{\mathbf{Z}}^{\dagger} \right) ^{-1}\tilde{\mathbf{Z}}\mathbf{P}_{\boldsymbol{\omega }_{\mathrm{D}}^{*}\left( f_{\mathrm{d}} \right)}^{}\tilde{\mathbf{Z}}^{\dagger}}{\tilde{\mathbf{a}}^{\dagger}\left( \mathbf{I}+\tilde{\mathbf{Z}}\mathbf{P}_{\boldsymbol{\omega }_{\mathrm{D}}^{*}\left( f_{\mathrm{d}} \right)}^{\bot}\tilde{\mathbf{Z}}^{\dagger} \right) ^{-1}\tilde{\mathbf{a}}} \right] 
					\\&
					\qquad\,  =\det \left[ \mathbf{I}+\tilde{\mathbf{Z}}\tilde{\mathbf{Z}}^{\dagger} \right] \det \left[ \mathbf{I}-\frac{\left( \mathbf{I}+\tilde{\mathbf{Z}}\tilde{\mathbf{Z}}^{\dagger} \right) ^{-1}\tilde{\mathbf{a}}\tilde{\mathbf{a}}^{\dagger}\left( \mathbf{I}+\tilde{\mathbf{Z}}\mathbf{P}_{\boldsymbol{\omega }_{\mathrm{D}}^{*}\left( f_{\mathrm{d}} \right)}^{\bot}\tilde{\mathbf{Z}}^{\dagger} \right) ^{-1}\tilde{\mathbf{Z}}\mathbf{P}_{\boldsymbol{\omega }_{\mathrm{D}}^{*}\left( f_{\mathrm{d}} \right)}^{}\tilde{\mathbf{Z}}^{\dagger}}{\tilde{\mathbf{a}}^{\dagger}\left( \mathbf{I}+\tilde{\mathbf{Z}}\mathbf{P}_{\boldsymbol{\omega }_{\mathrm{D}}^{*}\left( f_{\mathrm{d}} \right)}^{\bot}\tilde{\mathbf{Z}}^{\dagger} \right) ^{-1}\tilde{\mathbf{a}}} \right] 
					\\&
					\qquad\,    =\det \left[ \mathbf{I}+\tilde{\mathbf{Z}}\tilde{\mathbf{Z}}^{\dagger} \right] \left[ 1-\frac{\tilde{\mathbf{a}}^{\dagger}\left( \mathbf{I}+\tilde{\mathbf{Z}}\mathbf{P}_{\boldsymbol{\omega }_{\mathrm{D}}^{*}\left( f_{\mathrm{d}} \right)}^{\bot}\tilde{\mathbf{Z}}^{\dagger} \right) ^{-1}\tilde{\mathbf{Z}}\mathbf{P}_{\boldsymbol{\omega }_{\mathrm{D}}^{*}\left( f_{\mathrm{d}} \right)}^{}\tilde{\mathbf{Z}}^{\dagger}\left( \mathbf{I}+\tilde{\mathbf{Z}}\tilde{\mathbf{Z}}^{\dagger} \right) ^{-1}\tilde{\mathbf{a}}}{\tilde{\mathbf{a}}^{\dagger}\left( \mathbf{I}+\tilde{\mathbf{Z}}\mathbf{P}_{\boldsymbol{\omega }_{\mathrm{D}}^{*}\left( f_{\mathrm{d}} \right)}^{\bot}\tilde{\mathbf{Z}}^{\dagger} \right) ^{-1}\tilde{\mathbf{a}}} \right] 
					\\&
					\qquad\,     =\det \left[ \mathbf{I}+\tilde{\mathbf{Z}}\tilde{\mathbf{Z}}^{\dagger} \right] \frac{\tilde{\mathbf{a}}^{\dagger}\left( \mathbf{I}+\tilde{\mathbf{Z}}\tilde{\mathbf{Z}}^{\dagger} \right) ^{-1}\tilde{\mathbf{a}}}{\tilde{\mathbf{a}}^{\dagger}\left( \mathbf{I}+\tilde{\mathbf{Z}}\mathbf{P}_{\boldsymbol{\omega }_{\mathrm{D}}^{*}\left( f_{\mathrm{d}} \right)}^{\bot}\tilde{\mathbf{Z}}^{\dagger} \right) ^{-1}\tilde{\mathbf{a}}}.
			\end{split}}
		\end{equation}	     
	\end{figure*}  

	Similarly, we can obtain the MLE of $\mathbf{R}$ under the null hypothesis $\mathcal{H}_0$
	\begin{equation}
		\label{eq18}
		\hat{\mathbf{R}}_0=\frac{\mathbf{ZZ}_{}^{\dagger}+\mathbf{S}}{K\left( L+1 \right)}.
	\end{equation}
	Then, we have
	\begin{equation}
		\label{eq19}
		\begin{split}
			f\left( \mathbf{Z},\mathbf{Y}\left| \hat{\mathbf{R}}_0,\mathcal{H} _0 \right. \right) =&\frac{\mathrm{cont}.}{\det ^{K\left(L+1\right)}\left( \mathbf{ZZ}_{}^{\dagger}+\mathbf{S} \right)}
			\\
			=&\frac{\mathrm{cont}.}{\det ^{K\left(L+1\right)}\left( \mathbf{S} \right) \det ^{KL}\left( \tilde{\mathbf{Z}}\tilde{\mathbf{Z}}_{}^{\dagger}+\mathbf{I} \right)}.
		\end{split}
	\end{equation}
	Combining \eqref{eq9}, \eqref{eq270} and \eqref{eq19} yields  the OGLRT detector \eqref{eq200}.

	\bibliographystyle{IEEEtran}
	\bibliography{ref}
	
	%

	
	
	
	
	

\end{document}